\documentclass[11pt,a4paper]{article}
\usepackage{amsfonts,amssymb,mathrsfs,dsfont}
\usepackage{amsopn}
\usepackage{amsmath,amsthm}
\usepackage{tikz}
\usetikzlibrary{arrows}
\usepackage[square,comma,numbers,compress]{natbib}
\DeclareMathOperator{\curl}{\mathrm{curl}}
\DeclareMathOperator{\notni}{\reflectbox{$\notin$}}
\DeclareMathAlphabet{\mathpzc}{OT1}{pzc}{m}{it}

\newcommand\1{{\ensuremath {\mathds 1} }}
\newcommand{\C}{\mathbb{C}}

\newcommand{\Q}{\mathbb{Q}}
\newcommand{\R}{\mathbb{R}}
\newcommand{\Z}{\mathbb{Z}}
\newcommand{\bA}{\mathbf{A}}

\newcommand{\be}{\mathbf{e}}
\newcommand{\bJ}{\mathbf{J}}

\newcommand{\br}{\mathbf{r}}
\newcommand{\bx}{\mathbf{x}}

\newcommand{\by}{\mathbf{y}}

\newcommand{\sx}{\mathrm{x}}
\newcommand{\sz}{\mathrm{z}}

\newcommand{\cD}{\mathscr{D}} 
\newcommand{\cE}{\mathcal{E}}
\newcommand{\cG}{\mathcal{G}}

\newcommand{\cS}{\mathcal{S}}
\newcommand{\cV}{\mathcal{V}}

\newcommand{\sym}{\mathrm{sym}}
\newcommand{\asym}{\mathrm{asym}}
\newcommand{\loc}{\mathrm{loc}}
\newcommand{\nn}{\mathrm{nn}}
\newcommand{\domD}[2]{\cD^{#1}_{#2}}
\newcommand{\ordnu}{a}

\newcommand{\bDelta}{{\mbox{$\triangle$}\hspace{-8.0pt}\scalebox{0.8}{$\triangle$}}}

\theoremstyle{plain}
\newtheorem{thm}{Theorem}

\newtheorem{prop}[thm]{Proposition}

\theoremstyle{definition}

\theoremstyle{remark}

\newtheorem*{ack}{Acknowledgments}

\title{Many-anyon trial states}

\author{Douglas Lundholm}

\date{\scriptsize{Department of Mathematics, KTH Royal Institute of Technology\\ SE-100 44 Stockholm, Sweden}}

\begin{document}

\maketitle

\begin{abstract}
	The problem of bounding the (abelian) many-anyon ground-state energy from above,
	with a dependence on the statistics parameter which matches that
	of currently available lower bounds, is reduced to studying the correlation
	functions of Moore--Read (Pfaffian) and Read--Rezayi type clustering states.
\end{abstract}


\setlength\arraycolsep{2pt}
\def\arraystretch{1.2}

\section{Introduction} \label{sec:intro}

The importance of the concept of quantum statistics for our 
understanding of observed collective phenomena in nature cannot be overstated. 
While fermions stack according to 
the Pauli principle to form a Fermi sea, with its
implications for atomic structure, conduction bands in solids, etc.,
bosons can sit together to display amplified single-particle behavior,
as manifested by 
Bose--Einstein condensation and the coherent propagation of light.
Upon restricting to 
two spatial dimensions on the other hand, it turns out to be 
logically conceivable to have other types 
of identical particles than bosons and fermions,
satisfying braid statistics instead of permutation
statistics, and which have been given the name 
`anyons' 
\cite{LeiMyr-77, GolMenSha-81, Wilczek-82a, Wilczek-82b, Wu-84} (see also \cite[p.~386]{Souriau-70}).
These have the property that 
under \emph{continuous} exchange of two particles
their wave function changes not merely by a sign $\pm 1$,
but \emph{any} complex phase factor $e^{i\alpha\pi}$ is allowed, 
where the real number $\alpha$ is known as their `statistics parameter'.
Moreover, for logical consistency one has to keep track of any topological
winding of the particles during their exchange. 
For instance, if one particle moves around in configuration space in such a way
as to enclose $p$ other particles in a complete (counterclockwise, say) loop, 
a phase $2p\alpha\pi$ must arise,
while if two particles are exchanged once and 
in the process $p$ other particles are enclosed, 
the phase must be $(1+2p)\alpha\pi$.
All such topological complications vanish 
in the case of $\alpha=0$ (bosons) and $\alpha=1$ (fermions).
The concept has also been extended from phases (abelian) to unitary matrices
(non-abelian), 
but we shall here stick to the simpler 
(though not at all simple) abelian case.
Furthermore, instead of demanding that the wave function changes its phase
according to the above form of topological boundary conditions, 
also known as the `anyon gauge picture',
one may equivalently model such phases by means of attaching magnetic flux
to ordinary identical particles, i.e. bosons or fermions, 
resulting in a magnetic many-body interaction.
This is then called the `magnetic gauge picture' for anyons.
We refer to the extensive reviews 
\cite{DatMurVat-03, Forte-92, Froehlich-90, IenLec-92, Khare-05, Lerda-92, Myrheim-99, Ouvry-07, Stern-08, Wilczek-90}
for a more complete background on the topic.

The idea of particles with attached magnetic flux is fundamental to
the fractional quantum Hall effect (FQHE) 
\cite{Girvin-04, Goerbig-09, Jain-07, Laughlin-99, StoTsuGos-99}, 
which is a manifestation of a 
strongly correlated many-body state of electrons 
subject to planar confinement and a strong transverse magnetic field. 
More recently it has been proposed that similar effects 
apply to trapped bosonic atoms in artificial magnetic fields 
\cite{BloDalZwe-08, Cooper-08, MorFed-07, RonRizDal-11, Viefers-08},
and to graphene \cite{Bolotin-etal-09,Du-etal-09}.
The quasiparticles arising in the FQHE of electrons are predicted to be 
anyons with a corresponding fractional statistics parameter.
However, as discussed in \cite[Section~9.8.2]{Jain-07}, 
there has been some 
confusion in the literature concerning the exact values of $\alpha$ involved.
This can be traced to different conventions but also to the fact that
the statistics parameter has hitherto only been defined indirectly via the 
operation of adiabatic braiding and the computation of a corresponding Berry phase, 
as first outlined in \cite{AroSchWil-84}.
Only very recently has an effective Hamiltonian for anyons 
been derived \cite{LunRou-16},
which shows unambiguously how they may arise in a FQHE context and what
the statistics parameter then should be, confirming that fermions against the 
background of a fermionic Laughlin state with odd exponent $2p+1$ effectively 
couple to Laughlin quasiholes to 
behave as emergent anyons with $\alpha = 2p/(2p+1)$ 
(see also \cite{Rougerie-16}).

Despite the concept of anyons having been around now for almost four decades,
a satisfactory understanding of their physics is still lacking.
Due to their complicated many-body interaction (or geometry)
it has not been possible to 
solve the anyon Hamiltonian for its complete spectrum
or even its ground state, 
except in the two-particle case where it can be reduced to a one-particle problem
and thus be solved analytically \cite{LeiMyr-77, Wilczek-82b, AroSchWilZee-85},
while 
in the three- and four-particle cases it has been studied numerically 
\cite{SpoVerZah-91, MurLawBraBha-91, SpoVerZah-92, SpoVerZah-93}.
Nevertheless, as is evident from the vast body of literature 
(the author can count more than 300 papers on the topic),
there has been a fair amount of progress on the many-anyon problem,
most of which is based on various approximations. 
One of the most discussed is average-field theory 
(see e.g. \cite{ChenWilWitHal-89,Hosotani-93,IenLec-92,WenZee-90,Wilczek-90} for review),
where the individual anyons are replaced by their average magnetic field, 
something which is arguably reasonable in a sufficiently dense regime.
Other approximations assume either a very strong external magnetic field, 
thereby reducing to lowest-Landau-level anyons which turn out to be solvable 
with 
a connection to Calogero--Sutherland models \cite{DasOuv-94,Ouvry-07},
or in the case of the free dilute gas,
that it is sufficient
to only take two-particle interactions 
into account \cite{AroSchWilZee-85, Minor-93}.
It has however been stressed 
that real progress 
in understanding the anyon gas cannot be 
made without knowledge of the true many-body spectrum. 

In a recent series of works 
\cite{LunSol-13a,LunSol-13b,LunSol-14,LunRou-15,LarLun-16,CorLunRou-16}, 
the question concerning the many-anyon
ground state has been investigated in the light of modern mathematical methods. 
Interestingly, it was found that the ground-state energy for the free ideal anyon gas 
can be non-trivially bounded from below, but only under the assumption that $\alpha$
is an odd-\emph{numerator} rational number 
(in contrast to electron FQHE which typically involves odd-\emph{denominator}
filling factors).
To settle the issue whether this is the true behavior or rather just an artifact
of the method used to obtain the bounds, one also needs to bound 
the energy from above using suitable trial states. 
This however turns out to be a very difficult problem for anyons,
contrary to the more common situation where finding 
the upper bound is the easier part.

We shall here proceed in the setting of
abelian anyons with no 
external magnetic field
(which is indeed relevant in the FQHE context; cf. e.g. \cite{Stern-08}),
and propose,
building on \cite{LunSol-13b}, that good variational ground
states for the many-anyon problem are given by clustering states 
of the Moore--Read and Read--Rezayi type 
that have already been studied for some time 
in the context of special (proposedly non-abelian) regimes of the FQHE. 
In particular, these types of states seem to give a much lower 
energy for even-numerator $\alpha$ than for odd numerators,
and offer a corresponding picture of condensation, 
respectively, a reduced Fermi sea of anyons.
In view of the above considerations,
such a picture could in the context of the FQHE potentially have
far-reaching consequences.

\section{The many-anyon ground-state energy} \label{sec:energy}

For concreteness and for easier comparison with the results which are available 
in the literature, we consider as our starting point anyons confined in a 
harmonic oscillator potential.
In the magnetic gauge,
the Hamiltonian operator for $N$ non-relativistic ideal\footnote{That is,
without other interactions than the statistical one, and purely
pointlike as opposed to extended; cf. Section~\ref{sec:extended} below.} 
anyons with mass $m$ 
in a harmonic trap with frequency $\omega \ge 0$,
and in units such that $\hbar = 1$,
is
\begin{equation} \label{eq:Hamiltonian-osc}
	\hat{H}_N = \hat{T}_{\alpha} + \hat{V} 
	= \sum_{j=1}^N \left( \frac{1}{2m} D_j^2 + \frac{m\omega^2}{2} |\bx_j|^2 \right),
\end{equation}
where
$$
	D_j := -i\nabla_{\bx_j} + \alpha\bA_j(\bx_j)
$$
denotes the magnetically coupled momentum operator of particle $j$
at position $\bx_j \in \R^2$.
Each particle sees an Aharonov--Bohm magnetic flux $2\pi\alpha$ 
attached to every other particle, 
as is given explicitly by the magnetic potentials
\begin{equation} \label{eq:anyon-potential-ideal}
	\bA_j(\bx) := \sum_{\substack{k=1 \\ k \neq j}}^N (\bx - \bx_k)^{-\perp},
	\qquad \bx^{-\perp} := \frac{\bx^\perp}{|\bx|^2} = \frac{(-y,x)}{x^2 + y^2}.
\end{equation}
For reference we take the Hamiltonian $\hat{H}_N$ to act on bosonic states 
$\Psi \in L^2_\sym((\R^2)^N)$,
and there is associated with the free kinetic energy operator $\hat{T}_\alpha$
a natural subspace (form domain) $\domD{N}{\alpha}$ 
consisting of the states $\Psi$ which have a 
finite expectation value for their kinetic energy
(see \cite[Sec.~2.2]{LunSol-14} and \cite[Sec.~1.1]{LarLun-16} for details).
The case $\alpha=0$ then corresponds to bosons and $\alpha=1$ to fermions,
with 
$\domD{N}{\alpha=0} = H^1_\sym(\R^{2N})$ and 
$\domD{N}{\alpha=1} = U^{-1} H^1_\asym(\R^{2N})$
the Sobolev spaces of
symmetric/antisymmetric square-integrable functions having square-integrable
first derivatives.
In the latter case we have used the singular gauge transformation
$U^{-1}\nabla_{\bx_j}U = i\bA_j$, with
$$
	U\colon L^2_{\sym/\asym} \to L^2_{\asym/\sym}, 
	\qquad (U\Psi)(\sx) := \prod_{1 \le j<k \le N} \frac{z_j-z_k}{|z_j-z_k|} \Psi(\sx),
$$
and with the coordinates here represented by
$z_j = x_j+i y_j \in \C$,
to explicitly switch from 
fermions to their bosonic representation via flux attachment.
The same transformation can be used to show that the full spectrum must be 
periodic in $\alpha$ up to any even integer $2q$, 
by composing 
$\Psi$ with $U^{-2q}$ which preserves symmetry.
Also, one could equivalently have chosen to model everything in terms of fermions
with statistics parameter $\beta := \alpha-1$.

As is very well known, the harmonic oscillator
ground-state energy 
$$
	E_0(N) := \inf \textup{spec} \,\hat{H}_N 
	= \inf_{\Psi \in \domD{N}{\alpha} \setminus \{0\}} \frac{ \langle\Psi,\hat{H}_N\Psi\rangle }{ \langle\Psi,\Psi\rangle }
$$ 
is for bosons $E_0(N) = \omega N$, while for (spinless) fermions 
$E_0(N) \sim \frac{\sqrt{8}}{3}\omega N^{3/2}$
as $N \to \infty$ due to the Pauli principle and the filling of one-body states.
For fermions 
allowed to have $\nu \ge 1$ different spin states 
(or particles obeying Gentile intermediate statistics \cite{Gentile-40,Gentile-42})
it is a simple exercise to show using the same asymptotics that 
$E_0(N) \sim \frac{\sqrt{8}}{3}\nu^{-1/2}\omega N^{3/2}$.
However, for 
anyons with statistics parameter $\alpha$ it has now been established
\cite{LunSol-13a,LunSol-13b,LunSol-14,LarLun-16}
that 
\begin{equation} \label{eq:rigorous-bounds}
	C_1 \,j_{\alpha_N}' \,\omega N^{3/2}
	\ \le \ E_0(N) \ \le \ 
	C_2 \,\omega N^{3/2},
\end{equation}
for some universal constants $C_1 \le \sqrt{8}/(3j_1')$ 
and $C_2 \ge \sqrt{8}/3$.
Here $j_\ordnu'$ for $\ordnu>0$ denotes the first positive zero of 
the derivative of the Bessel function $J_\ordnu$ of the first kind, 
satisfying (see \cite{LarLun-16})
$$
	\sqrt{2\ordnu} \le j_\ordnu' \le \sqrt{2\ordnu(1+\ordnu)}
	\qquad \text{(and $j_0' := 0$)}.
$$
The order $\ordnu = \alpha_N$ involved is given by the `fractionality' 
of $\alpha$ as measured by
\begin{equation} \label{eq:alpha-N}
	\alpha_N 
		:= \min\limits_{p \in \{0,1,\ldots,N-2\}} \min\limits_{q \in \Z} |(2p+1)\alpha - 2q|.
\end{equation}
This expression has the peculiar many-body limit \cite[Prop.~5]{LunSol-13a}
(see Figure~\ref{fig:popcornplots})
$$
	\alpha_* := \lim_{N \to \infty} \alpha_N 
		= \inf_{N \ge 2} \alpha_N
		= \left\{ \begin{array}{ll}
		\frac{1}{\nu}, & \text{if $\alpha = \frac{\mu}{\nu} \in \Q$ reduced, $\mu$ \emph{odd} and $\nu \ge 1$,} \\
		0, & \text{otherwise.}
		\end{array}\right.
$$
In particular, the lower bound in \eqref{eq:rigorous-bounds} depends on $\alpha$ 
as $\sqrt{\alpha_*} = \nu^{-1/2}$ for small odd-numerator fractions
and tends to zero with $N$ for even-numerator and irrational numbers.

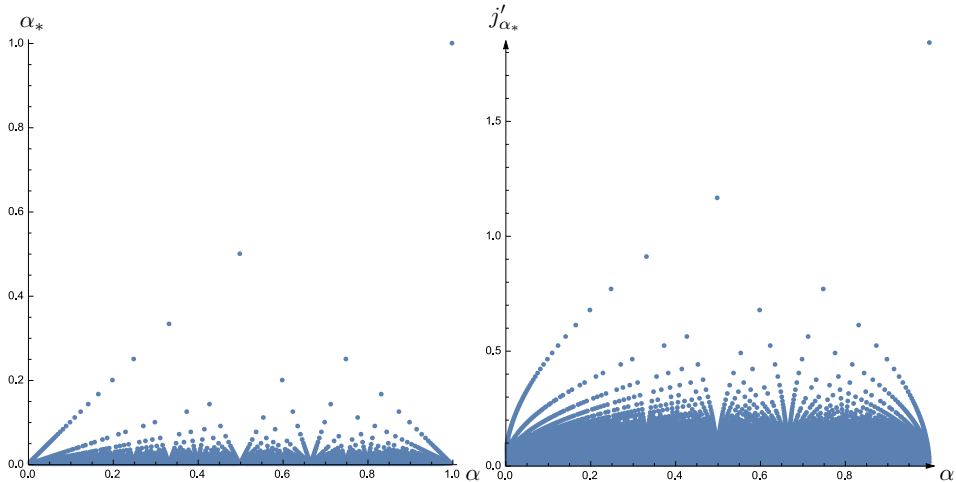
\begin{figure} 
	\centering
	\begin{tikzpicture}
		\node [above right] at (0,0) {\includegraphics[scale=0.65, trim=0.4cm 0cm 0cm 0cm]{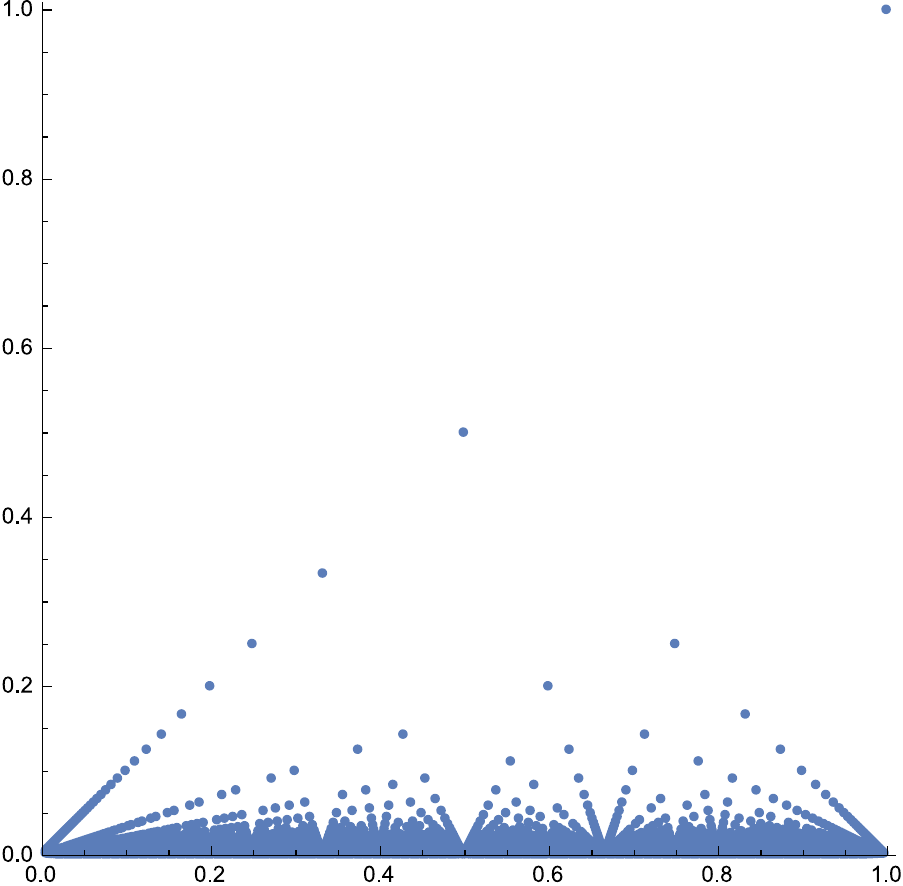}};
		\node [above right] at (5.70,0) {\scalebox{0.8}{$\alpha$\hspace{-20pt}}};
		\node [above right] at (-0.1,5.95) {\scalebox{0.8}{$\alpha_*$\hspace{-20pt}}};
	\end{tikzpicture}
	\begin{tikzpicture}
		\node [above right] at (0,0) {\includegraphics[scale=0.65, trim=0.1cm 0cm 0cm 0cm]{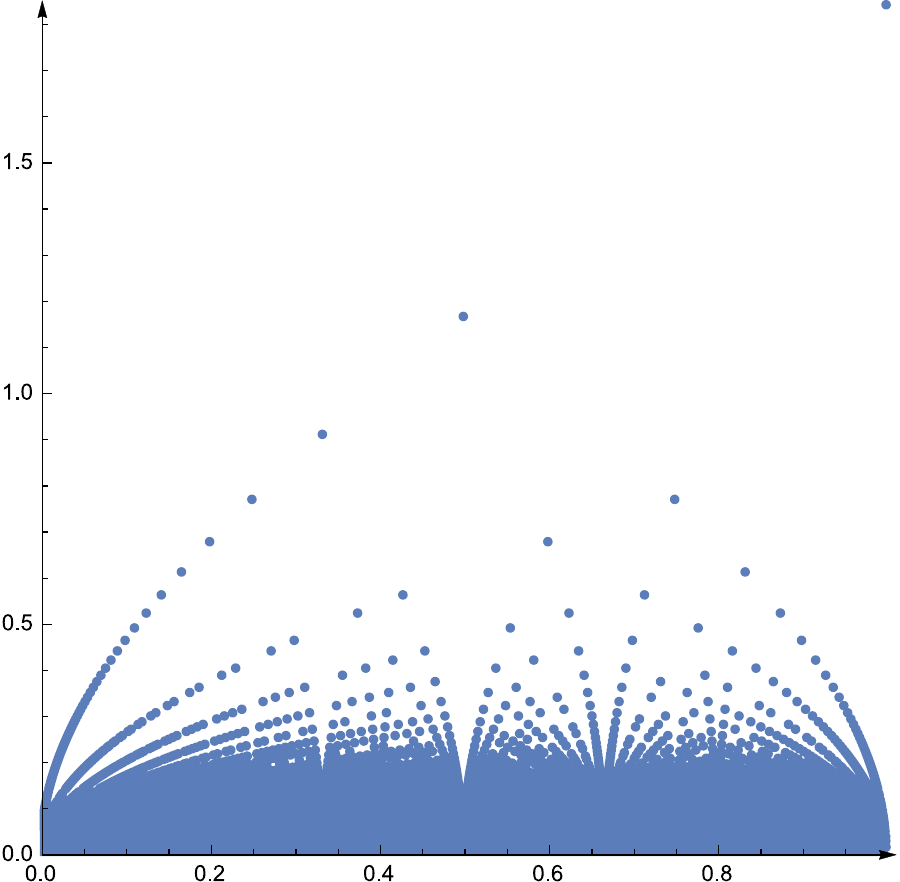}};
		\node [above right] at (5.85,0) {\scalebox{0.8}{$\alpha$\hspace{-20pt}}};
		\node [above right] at (0,5.95) {\scalebox{0.8}{$j_{\alpha_*}'$\hspace{-20pt}}};
	\end{tikzpicture}
	\caption{Plots of $\alpha_*$ respectively $j_{\alpha_*}'$ for 
		$0 \le \alpha \le 1$. These can be continued to all $\alpha \in \R$
		using periodicity and reflection (conjugation) symmetry.
		}
	\label{fig:popcornplots}
\end{figure}

In addition to the above, 
for an arbitrary state $\Psi$ with fixed total angular momentum $L \in \Z$
one also has the bound
\cite{ChiSen-92,Sen-92,Myrheim-99}
\begin{equation} \label{eq:CS-bound}
	\langle\Psi,\hat{H}_N\Psi\rangle  \ \ge \ 
	\omega \left( N + \left|L + \alpha \frac{N(N-1)}{2} \right| \right),
\end{equation}
where, if $L = -\alpha\binom{N}{2}$
(i.e. an average angular momentum of $-\alpha$ per particle pair)
would be achieved exactly on some state,
then the remaining bound is just the ground-state energy for bosons.
As discussed in \cite{ChiSen-92}, since $E_0(N) = O(N^{3/2})$ 
the bound \eqref{eq:CS-bound} implies that as the number
of particles increases there will be more and more level crossings
in the ground state to different angular momenta, resulting in 
a qualitative picture (see \cite[Fig.~1]{ChiSen-92})
with some interesting features in common with Figure~\ref{fig:popcornplots}.
Note that for every finite $N$, the ground-state curve is conjectured to be
continuous in $\alpha$ and piecewise smooth, however the 
number of such pieces grows like $\sqrt{N}$.

On the other hand, an approach that has been used extensively
in the literature to understand the anyon gas
(see the reviews and references therein) 
is to employ the `average-field approximation'\footnote{This 
is also known as the mean-field approximation in the literature, 
although it is useful to make
a distinction between the names;
cf. \cite{Wilczek-90,LunRou-15}.}
\begin{equation} \label{eq:avg-field}
	E_0(N) \approx \inf_{\substack{\varrho \ge 0 \\ \int\varrho = N}} \int_{\R^2} 
		\Biggl( \frac{\pi|\alpha|}{m} \varrho(\bx)^2 + 
		\frac{m\omega^2}{2}|\bx|^2 \varrho(\bx)
		\Biggr) d\bx
		= \frac{\sqrt{8}}{3} \sqrt{|\alpha|} \,\omega N^{3/2}.
\end{equation}
Here
it has been assumed that the anyons see each other via an approximately constant 
magnetic field $B(\bx) \sim 2\pi\alpha \varrho(\bx)$, 
with $\varrho(\bx)$ the local
density of particles, and hence they each have a lowest-Landau-level energy
$|B|/(2m) \sim \pi|\alpha|\varrho/m$.
In \cite{LunRou-15} it has been shown rigorously that such an approximation 
is indeed correct in the limit 
of almost-bosonic anyons 
(i.e. $\alpha \to 0$; see also \cite{CorLunRou-16}) in a confining trap,
however one needs to be careful with 
what is meant exactly with the approximation
and how such a limit is performed since strictly speaking the anyons cannot be
ideal but extended (see \cite{Trugenberger-92b} and Section~\ref{sec:extended} below).
Also note that the periodicity for ideal anyons mentioned above is not 
naturally implemented in \eqref{eq:avg-field}, 
so we must at least expect to replace $\alpha$ with its periodization 
$\alpha_2$ from \eqref{eq:alpha-N}.

The main question raised from the bounds \eqref{eq:rigorous-bounds} and 
\eqref{eq:CS-bound} is whether for certain
$\alpha$ such that $\alpha_* \ll \alpha_2$
--- most notably for even-numerator rational $\alpha$ 
such as $\alpha=2/3$, for which $\alpha_*=0$ ---
and for particular states such that $L \sim -\alpha\binom{N}{2}$,
the true ground-state energy $E_0(N)$ could be considerably lower than the one 
\eqref{eq:avg-field} expected from average-field theory.
Building on \cite{LunSol-13b},
we shall here explore the possibility that this is actually the case.

\section{A local exclusion principle for anyons} \label{sec:exclusion}

In order to understand the origin of the peculiar dependence of the above 
energy bounds on $\alpha_*$, and of the form of the corresponding proposed 
trial states, we first need to briefly discuss the findings in 
\cite{LunSol-13a,LunSol-13b,LarLun-16} 
of a local exclusion principle for anyons.

Normally when one talks about an exclusion principle for identical particles
one has in mind an occupation picture, where only
a limited number of particles can sit in each distinguishable one-body state.
The prime example is of course the usual Pauli exclusion principle for fermions,
although various extensions have also been discussed in the literature 
\cite{Gentile-40,Gentile-42,Haldane-91,DasOuv-94,Isakov-94,Wu-94,HanLeiVie-01}
(the role of such generalized exclusion in the context of anyons has
also been reviewed in \cite{CanJoh-94}).
Sometimes the notion is generously extended to concern exclusion of 
coincident points in the configuration space,
and anyons have often been pointed out to obey such exclusion, 
either because it is required for their topological definition or because
their singular
(for ideal anyons) magnetic interaction forces the wave function to vanish
on the diagonals, 
just like fermions do by means of antisymmetry (see also \cite{BorSor-92}).
However, note that such a notion of exclusion also applies to the hard-core 
Bose gas whose thermodynamic ground-state energy
vanishes in the dilute limit 
(in analogy with the non-interacting gas)
in two and higher dimensions
\cite{LeeYan-57,Dyson-57,Schick-71,LieYng-98,LieYng-01,LieSeiSolYng-05}.

For anyons whose statistics is generated by a true many-body interaction 
(or a very complicated geometry),
a stronger notion of exclusion is required, and such a notion has been developed 
in \cite{LunSol-13a,LunSol-13b,LarLun-16} in the form of an effective 
repulsive long-range pair interaction.
Namely, recall that
the effect of the statistics is a change of phase of the anyonic
wave function by $e^{i\alpha\pi}$ whenever two particles are interchanged via a 
simple continuous loop in configuration space,
or a total phase $(1+2p)\alpha\pi$ whenever 
such an exchange loop at the same time encloses $p$ other particles.
On the other hand, particles are also allowed to have pairwise relative angular
momenta, 
and such momenta are by continuity restricted to only change the phase by an 
\emph{even} multiple of $\pi$.
One way of viewing this condition (in the magnetic gauge picture) 
is that the many-body wave function is modeled as a
bosonic (or fermionic) one and thus it needs to be antipodal-symmetric
(resp. antisymmetric) with respect to the relative coordinate, hence
$\pi$-periodic (resp. anti-periodic) in the relative angle.
Another way to see it (in the anyon gauge or geometric picture) is that the wave function
is a section of a locally flat complex line bundle 
with the topological continuity condition
that its phase around such a loop should jump by 
$(1+2p)\alpha\pi$ plus $2\pi$ times an arbitrary winding number.
Assuming then that the relative momentum or winding is the even
integer $-2q$ if the particles 
are orbiting in a reversed (for $\alpha>0$) direction in order to cancel 
as much of the magnetic or topological phase as possible, 
we arrive at a total phase
$(2p+1)\alpha - 2q$ times $\pi$ per exchange for the particle pair.
In the kinetic energy,
any nonzero remainder 
phase of this sort 
gives rise to a centrifugal-barrier 
repulsion
\begin{equation} \label{eq:statistical-repulsion}
	V_{\textrm{stat}}(r) = |(2p+1)\alpha - 2q|^2\frac{1}{r^2} 
	\ \ge \ \frac{\alpha_N^2}{r^2}
	\ \ge \ \frac{\alpha_*^2}{r^2},
\end{equation}
where $r$ denotes the relative distance of the particle pair,
and we have taken the infimum over all possibilities
$p \in \{0,1,\ldots,N-2\}$ and $q \in \Z$ to obtain the lower bounds in terms of 
$\alpha_N \ge \alpha_*$. 
In the case of $\alpha = \mu/\nu$ being an arbitrary reduced fraction 
with an \emph{odd} numerator $\mu$ and a positive denominator $\nu$, 
it turns out using simple number theory \cite[Prop.~5]{LunSol-13a}
that this phase mismatch can never be completely cancelled, 
and in fact $\alpha_* = 1/\nu > 0$.
On the other hand, if $\mu$ is an \emph{even} number it is evident 
that cancellation
is possible for particular values of $p$ and $q$, and therefore $\alpha_* = 0$.
For irrational values of $\alpha$, one can use that any such number can be 
approximated arbitrarily well by both even- and odd-numerator rational numbers 
and hence $\alpha_* = 0$ 
(although note that a very large $N$ may be required in such a process).

A geometric interpretation of the potential \eqref{eq:statistical-repulsion}
is that there is non-trivial curvature (magnetic flux)
sitting at each of the enclosed particles
but effectively seen at the center $r=0$ of the particle pair,
and its presence is felt by the kinetic energy in the form of an
effective potential.
The situation is from this perspective 
analogous to that of a free quantum particle moving on a cone 
\cite{tHooft-88,DesJac-88,KayStu-91}, here
with its apex angle depending on $\alpha$ and the number of enclosed particles.

In \cite{LunSol-13a} and \cite{LarLun-16} it has been shown rigorously by
means of a family of magnetic Hardy inequalities that such an effective pairwise 
inverse-square `statistical repulsion' \eqref{eq:statistical-repulsion}
indeed arises in the many-anyon system. 
Although the effect is in some sense local and weighted only linearly in the 
number of particles 
(in contrast to a usual pair-interaction term in the Hamiltonian), 
it is still of long-range type 
and sufficiently strong to produce a `degeneracy pressure' 
(represented concretely in the form of Lieb--Thirring inequalities;
cf. \cite{LieThi-75,LunSol-13b})
and consequently non-trivial energy bounds in terms of $\alpha_*$
for the ideal or dilute anyon gas, such as \eqref{eq:rigorous-bounds}.
We also stress that the method 
used to obtain the effective pair potential \eqref{eq:statistical-repulsion}, 
which was 
introduced in \cite{LunSol-13a} and developed to encompass more general 
situations in \cite{LarLun-16},
is well suited for numerical investigations of lower bounds to 
the ground-state energy.

\section{Constructing anyonic trial states} \label{sec:trial}

In order to match the available lower bounds for $E_0(N)$ from above,
it was in \cite{LunSol-13b}
suggested to study variational trial states of the form
$$
	\Psi = \Phi \psi_\alpha \ \in L^2_{\textup{sym}}(\R^{2N}),
$$
with $\Phi \in L^1_{\loc,\sym}$ a locally integrable regularizing factor, 
$N = \nu K$ a suitable sequence of particle numbers and, 
in the case of $\alpha$ being 
an even-numerator reduced fraction $\alpha = \mu/\nu \in [0,1]$, 
\begin{align} \label{eq:trial-even}
	\psi_\alpha &:= \prod_{j<k} |z_{jk}|^{-\alpha} 
		\,\cS\left[ \prod_{q=1}^\nu 
		\prod_{(j,k) \in \cE_q} (\bar{z}_{jk})^\mu 
		\right] \prod_{l=1}^N \varphi_0(\bx_l), 
\end{align}
while for odd numerators $\mu$,
\begin{equation} \label{eq:trial-odd}
	\psi_\alpha := \prod_{j<k} |z_{jk}|^{-\alpha} 
		\,\cS\left[ \prod_{q=1}^\nu 
		\prod_{(j,k) \in \cE_q} (\bar{z}_{jk})^\mu 
		\bigwedge_{k=0}^{K-1} \varphi_k \,(\bx_{l \in \cV_q})
		\right].
\end{equation}
Here $z_{jk} := z_j - z_k$ are the pairwise relative complex coordinates
with the usual identification $\C \ni z_j \leftrightarrow \bx_j \in \R^2$,
and we have grouped, or `colored', the particles into $\nu$ different colors 
where $G_q = (\cV_q,\cE_q)$ denotes the complete graph 
over each such group of $|\cV_q| = K$ vertices=particles (cf. Figure~\ref{fig:clustering}).
The symmetrization $\cS$ over all the particles then amounts to 
symmetrization over all such colorings,
and can be viewed as passing from a set of distinguishable particles (by color)
to indistinguishable (cf. \cite{RegGoeJol-08}).
The $\varphi_k$, $k=0,1,2,\ldots$, are the 
eigenstates (ordered by increasing energy) 
of a corresponding one-body Hamiltonian
\begin{equation} \label{eq:Hamiltonian_1}
	\hat{H}_1 = \frac{1}{2m}
		\bigl(-i\nabla_\bx + \bA_\textrm{ext}(\bx)\bigr)^2 + V(\bx), 
\end{equation}
and in \eqref{eq:trial-odd} we have formed the Slater determinant of the $K$ 
first such states in the variables of each color group $\cV_q$
to obtain matching symmetry.

\begin{figure}
	\centering
	\begin{tikzpicture}[>=stealth']
		\def\x{0}
		\def\y{0}
		\def\c{red}
		\coordinate (p1) at (0+\x,0+\y);
		\coordinate (p2) at (0+\x,1+\y);
		\coordinate (p3) at (1+\x,1+\y);
		\coordinate (p4) at (1+\x,0+\y);
		\draw [fill,\c] (p1) circle [radius = 0.04];
		\draw [fill,\c] (p2) circle [radius = 0.04];
		\draw [fill,\c] (p3) circle [radius = 0.04];
		\draw [fill,\c] (p4) circle [radius = 0.04];
		\draw [\c] (p1) -- (p2) -- (p3) -- (p4) -- (p1);
		\draw [\c] (p1) -- (p3);
		\draw [\c] (p2) -- (p4);

		\def\x{2}
		\def\y{0}
		\def\c{green}
		\coordinate (p1) at (0+\x,0+\y);
		\coordinate (p2) at (0+\x,1+\y);
		\coordinate (p3) at (1+\x,1+\y);
		\coordinate (p4) at (1+\x,0+\y);
		\draw [fill,\c] (p1) circle [radius = 0.04];
		\draw [fill,\c] (p2) circle [radius = 0.04];
		\draw [fill,\c] (p3) circle [radius = 0.04];
		\draw [fill,\c] (p4) circle [radius = 0.04];
		\draw [\c] (p1) -- (p2) -- (p3) -- (p4) -- (p1);
		\draw [\c] (p1) -- (p3);
		\draw [\c] (p2) -- (p4);
		
		\def\x{1}
		\def\y{2}
		\def\c{blue}
		\coordinate (p1) at (0+\x,0+\y);
		\coordinate (p2) at (0+\x,1+\y);
		\coordinate (p3) at (1+\x,1+\y);
		\coordinate (p4) at (1+\x,0+\y);
		\draw [fill,\c] (p1) circle [radius = 0.04];
		\draw [fill,\c] (p2) circle [radius = 0.04];
		\draw [fill,\c] (p3) circle [radius = 0.04];
		\draw [fill,\c] (p4) circle [radius = 0.04];
		\draw [\c] (p1) -- (p2) -- (p3) -- (p4) -- (p1);
		\draw [\c] (p1) -- (p3);
		\draw [\c] (p2) -- (p4);
		
		\node [below right] at (0,0) {\scalebox{0.9}{$\cV_1$}};
		\node [below right] at (2,0) {\scalebox{0.9}{$\cV_2$}};
		\node [below right] at (1,2) {\scalebox{0.9}{$\cV_3$}};
		\node [below right] at (0.95,0.2) {\scalebox{0.9}{$j$}};
		
		\node [above right] at (4.0,1.2) {$\Rightarrow$};

		\def\x{6}
		\def\y{0}
		\draw [arrows=->,thick,] 	(3.2+\x,0.3+\y) -- (3.7+\x,0.95+\y);
		\draw [arrows=->,thick,] 	(3.2+\x,0.3+\y) -- (2.7+\x,-0.3+\y);
		\def\c{red}
		\coordinate (p1) at (0.3+\x,0.2+\y);
		\coordinate (p2) at (0.0+\x,3.2+\y);
		\coordinate (p3) at (2.6+\x,2.9+\y);
		\coordinate (p4) at (3.2+\x,0.3+\y);
		\draw [fill,\c] (p1) circle [radius = 0.04];
		\draw [fill,\c] (p2) circle [radius = 0.04];
		\draw [fill,\c] (p3) circle [radius = 0.04];
		\draw [fill,\c] (p4) circle [radius = 0.04];
		\draw [\c] (p1) -- (p2) -- (p3) -- (p4) -- (p1);
		\draw [\c] (p1) -- (p3);
		\draw [\c] (p2) -- (p4);
		\def\c{green}
		\coordinate (p1) at (0.1+\x,0.1+\y);
		\coordinate (p2) at (0.3+\x,2.8+\y);
		\coordinate (p3) at (2.4+\x,3.2+\y);
		\coordinate (p4) at (2.8+\x,0.1+\y);
		\draw [fill,\c] (p1) circle [radius = 0.04];
		\draw [fill,\c] (p2) circle [radius = 0.04];
		\draw [fill,\c] (p3) circle [radius = 0.04];
		\draw [fill,\c] (p4) circle [radius = 0.04];
		\draw [\c] (p1) -- (p2) -- (p3) -- (p4) -- (p1);
		\draw [\c] (p1) -- (p3);
		\draw [\c] (p2) -- (p4);
		\def\c{blue}
		\coordinate (p1) at (0.2+\x,0.5+\y);
		\coordinate (p2) at (0.1+\x,3.0+\y);
		\coordinate (p3) at (3.2+\x,3.1+\y);
		\coordinate (p4) at (2.6+\x,0.4+\y);
		\draw [fill,\c] (p1) circle [radius = 0.04];
		\draw [fill,\c] (p2) circle [radius = 0.04];
		\draw [fill,\c] (p3) circle [radius = 0.04];
		\draw [fill,\c] (p4) circle [radius = 0.04];
		\draw [\c] (p1) -- (p2) -- (p3) -- (p4) -- (p1);
		\draw [\c] (p1) -- (p3);
		\draw [\c] (p2) -- (p4);

		\draw [dashed] (0.2+\x,0.35+\y) circle [radius = 0.50];
		\draw [dashed] (0.1+\x,3.0+\y) circle [radius = 0.50];
		\draw [dashed] (2.8+\x,3.15+\y) circle [radius = 0.50];
		\draw [dashed] (2.85+\x,0.3+\y) circle [radius = 0.50];

		\node [below right] at (3.3+\x,0.5+\y) {\scalebox{0.9}{$\bx_j$}};
		\node [below right] at (3.3+\x,1.45+\y) {\scalebox{0.9}{$\alpha\bA_j$}};
		\node [below right] at (2.2+\x,-0.1+\y) {\scalebox{0.9}{$\bJ_j$}};
		\node [below left] at (-0.2+\x,2.8+\y) {\scalebox{0.9}{$\cV_1^*$}};
	\end{tikzpicture}
	\caption{An illustration of a coloring of $N=12$ particles with
		$\nu=3$ colors into $K=4$ clusters. Each colored edge corresponds
		to one unit $-\mu$ of pairwise angular momentum.
		Also shown is the contribution to the magnetic potential
		$\alpha\bA_j$ and the current $\bJ_j$ of particle $j$ due solely 
		to the cluster $\cV_1^*$.}
	\label{fig:clustering}
\end{figure}
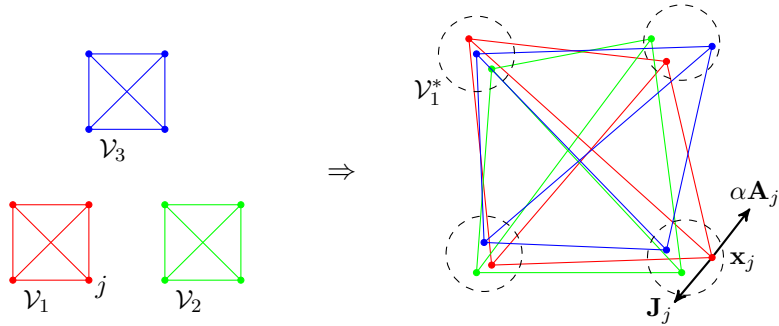

In the harmonic oscillator case 
$\bA_\textrm{ext} = 0$, $V(\bx) = m\omega^2|\bx|^2/2$,
these $N$-body states have angular momentum
\begin{equation} \label{eq:trial-momentum}
	L = -\mu\nu\binom{K}{2}
	= -\alpha\binom{N}{2} + \alpha\frac{\nu-1}{2}N,
\end{equation}
for \eqref{eq:trial-even} and for certain `magic' numbers 
$K$ in \eqref{eq:trial-odd} corresponding to filled shells.
The symmetrized quantity in \eqref{eq:trial-even}
is then a homogeneous polynomial in $\bar{z}_j$ which, 
multiplied with the Gaussian $e^{-m\omega|\sz|^2/2}$, 
$|\sz|^2 = \sum_{j=1}^N |z_j|^2$,
coincides exactly with the (complex-conjugated bosonic)
Laughlin--Read--Rezayi states for the fractional quantum Hall effect
\cite{ReaRez-99,CapGeoTod-01,RegGoeJol-08}.
Also note that, if the $\varphi_k$ are instead taken to be the
lowest-Landau-level states of a constant magnetic field with cyclotron
frequency $\omega$ in symmetric gauge,
$$
	\varphi_k(z) \propto \bar z^k e^{-m\omega |z|^2/2},
	\quad k = 0,1,2,\ldots,
$$
then the case $\alpha=1/2$ in \eqref{eq:trial-odd} corresponds to the
Moore--Read (Pfaffian) states \cite{MooRea-91,CapGeoTod-01}
(cf. also \cite{Girvin-etal-90,WenZee-90} where pairing of semions has
been discussed).

It is known that these states possess clustering properties, so that 
for example the symmetric polynomial\footnote{The normalization factor here is 
chosen to simplify the identities below. The number of terms in the 
symmetrized expression is $(\nu K)!/(\nu! (K!)^\nu)$.}
$$
	f_{N = \nu K}(\sz) := \frac{1}{(\nu!)^{K-1}} \,
	\cS\left[ \prod_{q=1}^\nu 
		\prod_{(j,k) \in \cE_q} (z_{jk})^\mu
		\right],
	\quad \mu = 2,4,6,\ldots,
$$
(together with a confining factor such as the Gaussian)
exhibits clusters of $\nu$ particles. 
This, as well as other interesting and useful properties of such symmetric 
polynomials, have also been observed by means of an identification
with correlators of certain conformal field theories 
(see e.g. \cite{HanHerSimVie-16,ArdKedSto-05})
and with Jack polynomials \cite{BerHal-08}.
For instance, one has that if the positions of $\nu$ particles (i.e. one cluster)
are identified, then 
(also compare to Figure~\ref{fig:clustering})
\begin{equation} \label{eq:Jack-clustering}
	f_N(\underbrace{\zeta,\ldots,\zeta}_{\nu\ \text{copies}},z_{\nu+1},z_{\nu+2},\ldots,z_N)
	= \prod_{j=\nu+1}^{N} (\zeta - z_j)^\mu f_{N-\nu}(z_{\nu+1},\ldots,z_N).
\end{equation}
In particular, $f_N$ then vanishes whenever $\nu+1$ or more particles 
are brought together.
Furthermore, if proceeding in this way to group all particles into 
disjoint clusters $\cV_q^*$, $q=1,\ldots,K$, with $|\cV_q^*| = \nu$
(think of complete graphs $G_q^* = (\cV_q^*,\cE_q^*)$ dual to $G_q$ in a sense), 
and then identifying their positions,
say $z_j = \zeta_q$ for $j \in \cV_q^*$, 
one obtains a Laughlin state with exponent $\nu^2\alpha$,
\begin{equation} \label{eq:Jack-Laughlin}
	f_N(\sz) = \prod_{1 \le p < q \le K} (\zeta_p-\zeta_q)^{\nu\mu}.
\end{equation}

Note that this clustering behavior matches very well with both the form of the
magnetic potential $\bA_j$ and the Jastrow prefactor in $\psi_\alpha$. Namely,
the attractive Jastrow factor contracts the clusters and balances with the 
inter-cluster repulsion coming from the Jack polynomial $f_N(\bar\sz)$ 
in such a way that each particle $\bx_j$ sees from any other
cluster $\cV_q^* \notni j$, say located at $\by \leftrightarrow \zeta$ 
at a large distance $r = |\br| = |z_j-\zeta|$, 
the attraction 
$\prod_{k \in \cV_q^*} |z_{jk}|^{-\alpha} \sim r^{-\nu\alpha} = r^{-\mu}$ 
and at the same time the effective repulsion $\sim |z_j - \zeta|^\mu = r^\mu$ 
from \eqref{eq:Jack-clustering}.
Also, the total contribution to the magnetic potential seen by particle $\bx_j$
from this cluster is 
$\alpha\bA_j(\bx_j; \bx_{k \in \cV_q^*}) \sim \nu\alpha \br^{-\perp}$,
while at the same time the particle has an orbital angular momentum
around the cluster with an opposite current contribution 
$\bJ_j = -i\psi_\alpha^{-1}\nabla_j \psi_\alpha(\bx_j; \bx_{k \in \cV_q^*}) \sim -\mu \br^{-\perp}$, 
again thanks to \eqref{eq:Jack-clustering} (and complex conjugation).
One should also observe that (see Figure~\ref{fig:clustering}), 
due to the balance between Jastrow attraction and 
Jack repulsion, clusters are formed out 
of particles with different colors, i.e. in different groups $\cV_q$.
Furthermore, every particle has exactly one edge in $\cE_q$
going to exactly one particle in every other cluster, 
namely the particle of the same color,
and this is what gives the orbital angular momentum contribution 
$(\bar{z}_{jk})^\mu$.
In this way there is a natural cancellation between magnetic flux and
angular momentum on the level of each individual particle.
The same holds 
in the case of the odd-numerator states, which however have 
an additional
repulsion and possibly angular momentum coming from the Slater determinant
in \eqref{eq:trial-odd}.

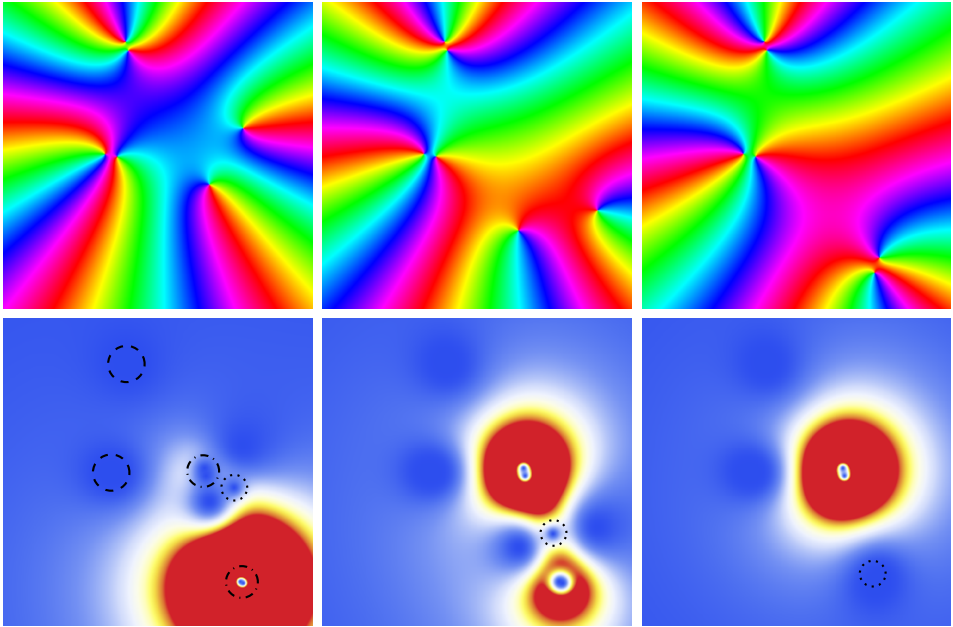
\begin{figure}[t]
	\centering
	\begin{tikzpicture}
		\node [above right] at (0  ,4.2) {\includegraphics[scale=0.32]{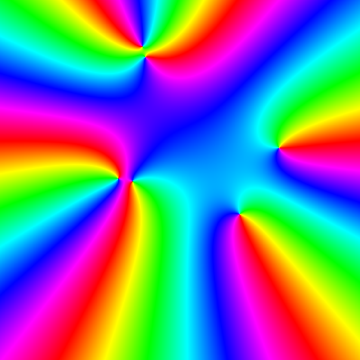}};
		\node [above right] at (4.2,4.2) {\includegraphics[scale=0.32]{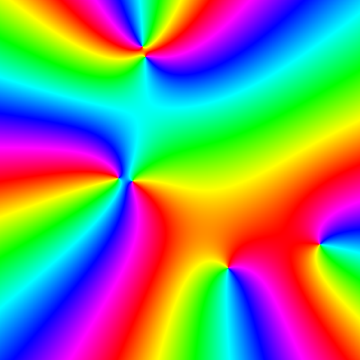}};
		\node [above right] at (8.4,4.2) {\includegraphics[scale=0.32]{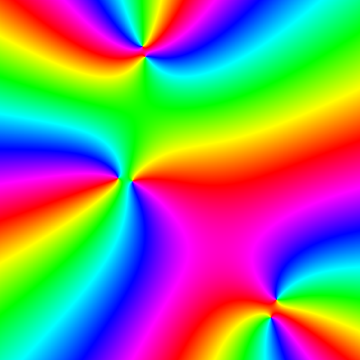}};
		\node [above right] at (0  ,0) {\includegraphics[scale=0.32]{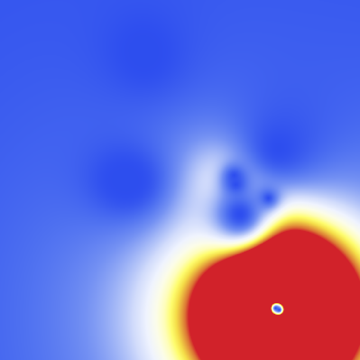}};
		\node [above right] at (4.2,0) {\includegraphics[scale=0.32]{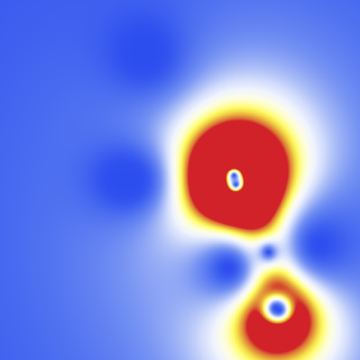}};
		\node [above right] at (8.4,0) {\includegraphics[scale=0.32]{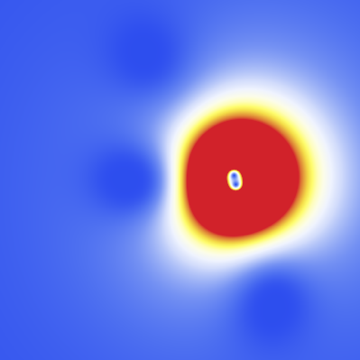}};
		\draw [dashed,thick] (1.75,3.60) circle [radius = 0.24];
		\draw [dashed,thick] (1.55,2.16) circle [radius = 0.24];
		\draw [dashdotted,thick] (2.76,2.18) circle [radius = 0.21];
		\draw [dashdotted,thick] (3.27,0.71) circle [radius = 0.21];
		\draw [dotted,thick] (		3.17,1.96) circle [radius = 0.17];
		\draw [dotted,thick] (4.2+	3.17,1.36) circle [radius = 0.17];
		\draw [dotted,thick] (8.4+	3.17,0.82) circle [radius = 0.17];
	\end{tikzpicture}
	\caption{%
		A sequence of three particle configurations chosen to illustrate the properties
		of the trial states $\Psi$, here with $\alpha=2/3$, $N=12$,
		and $\Phi = \Phi_{r_0 = 1.3}$ as defined in \eqref{eq:Phi-parameter}. 
		We fix the positions of 11 particles in two 3-clusters (dashed circles),
		two 2-clusters (dash-dotted circles),
		and a single particle at three intermediate positions 
		between the 2-clusters (dotted circles).
		Then $\arg \Psi$ (top row) and $|\Psi|^2$ (bottom row)
		are plotted as functions of the remaining particle position
		on the square $[-10,10]^2$.
		One may note that each 3-cluster binds two vortices, 
		which produce the desired orbital angular momentum
		and also cancel the attraction from the cluster very effectively.
		The attraction is much stronger from a 2-cluster than from a single particle,
		but entirely absent from 3-clusters.
		}
	\label{fig:clusterplots-1p}
\end{figure}

In the fully clustered picture \eqref{eq:Jack-Laughlin},
\begin{equation} \label{eq:cluster-gauge}
	\Psi = \Phi\psi_\alpha \sim \prod_{1 \le p < q \le K} 
		\left( \frac{ \overline{\zeta_p-\zeta_q} }{|\zeta_p-\zeta_q|} \right)^{\nu\mu}
\end{equation}
becomes the necessary gauge transformation $U^{-\nu\mu}$
to remove the overall statistical
effects of the clusters 
(here we have $\nu$ copies of the even integer $\mu$ 
because there are $\nu$ particles moving in each cluster,
each seeing a flux $\nu\alpha=\mu$ from every other cluster).
Note that the role of $\Phi$ is to regularize the singular short-scale dependence 
of $\psi_\alpha$ arising upon bringing particles very close together
(the Jastrow factor in $\psi_\alpha$
diverges with each pair like $|z_{jk}|^{-\alpha}$),
and in \eqref{eq:cluster-gauge} we have assumed that this has effectively 
removed the singular factor
$\prod_{(j,k) \in \cE_q^*} |z_{jk}|^{-\alpha}$ from within each cluster.
We also note that the states \eqref{eq:trial-even} and \eqref{eq:trial-odd}
naturally generalize for $\alpha \in \Z$ to the correct gauge copies 
$\Psi = U^{-\alpha} \Psi_0$
of the bosonic resp. fermionic ground states, 
$\Psi_0 = \otimes^N \varphi_0$ resp. $\Psi_0 = \wedge_{k=0}^{N-1} \varphi_k$,
in this case leaving out the need for the regulator $\Phi$.

With the assumption that the total energy increases with the 
total degree of one-particle
states $\varphi_k$, the odd-numerator states clearly have a higher energy
than the even-numerator ones, simply enforced by the symmetry constraint.
One needs to explain, however, why one cannot just take the same states 
$\psi_\alpha$ but shifted to a reducible fraction
$\alpha = \mu/\nu = \mu'/\nu'$, with $\mu' = k\mu$, $\nu' = k\nu$, $k \ge 2$.
Note first that the necessary properties of the states may not be valid for
such fractions and indeed 
there are certain assumptions on irreducibility in 
the context of Jack polynomials \cite{BerHal-08},
but let us proceed anyway with the discussion,  
aiming for a better understanding of the argument.

Within the class of even or odd states we expect that the size of clusters,
i.e. the denominator $\nu$ resp. $\nu'$, dictates the energy of the 
regulator $\Phi$ which therefore favors the irreducible case $k=1$.
Also note that we cannot shift from an even state to an odd one by 
extending the fraction by $k$ in this way, 
but one could certainly take $k$ to be even and thereby shift an odd state into 
an even one, and thus argue that the energy should then become lower
(consider for example $\alpha = 2/6$ instead of $1/3$ 
or, for $\alpha=1$, 
Cooper pairs instead of a single Slater determinant of fermions).
To argue against this possibility we need to study the pairwise 
structure of the states closer.

\begin{figure}[t]
	\centering
	\begin{tikzpicture}
		\node [above right] at (0  ,4.2) {\includegraphics[scale=0.32]{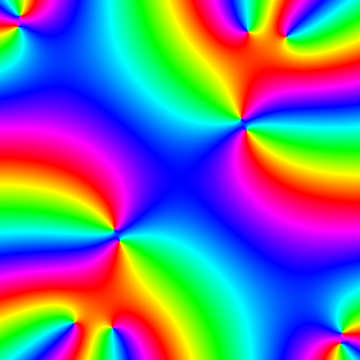}};
		\node [above right] at (4.2,4.2) {\includegraphics[scale=0.32]{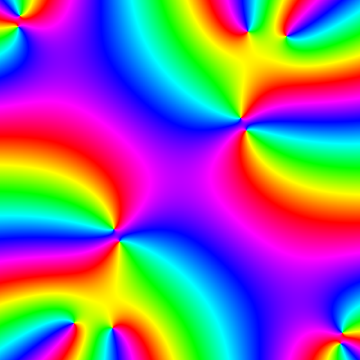}};
		\node [above right] at (8.4,4.2) {\includegraphics[scale=0.32]{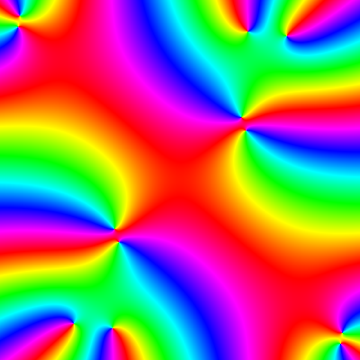}};
		\node [above right] at (0  ,0) {\includegraphics[scale=0.32]{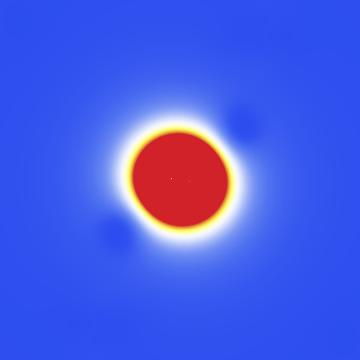}};
		\node [above right] at (4.2,0) {\includegraphics[scale=0.32]{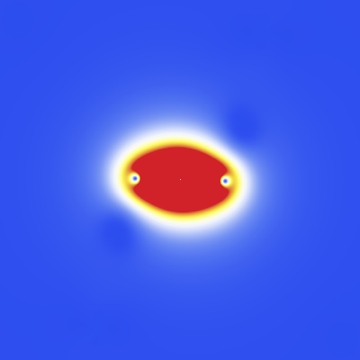}};
		\node [above right] at (8.4,0) {\includegraphics[scale=0.32]{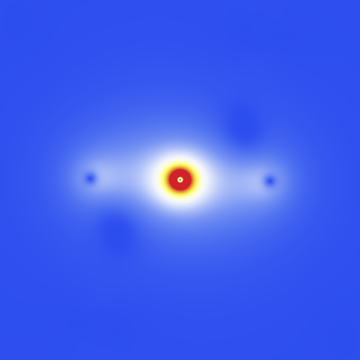}};
		\draw [dotted,thick] (2.17,6.38) circle [radius = 0.95];
	\end{tikzpicture}
	\caption{%
		A similar plot as in Figure~\ref{fig:clusterplots-1p} 
		with the same state $\Psi$.
		Now three 3-clusters are fixed, 
		one particle takes three different positions along the horizontal axis,
		while the relative coordinate of the remaining pair of particles 
		is plotted on the square $[-20,20]^2$ with the center of mass
		of the pair fixed at the origin.
		In relative coordinates the picture is 
		antipodal-symmetrized,
		and one may also note that on a circle with radius equal to the nearest
		3-cluster distance (dotted) the total phase circulation is $4\pi$.
		In this way the probability of finding a third particle within
		an interparticle distance may possibly be increased without a 
		significant cost in exclusion energy.
		}
	\label{fig:clusterplots-2p}
\end{figure}

First note that the pairwise relative angular momentum for any 
(shifted or not) even state
$\psi_{\alpha = \mu'/\nu'}$ always comes in multiples of $\mu'$.
The statistical repulsion \eqref{eq:statistical-repulsion} 
then always gives a positive energy unless for some 
$p,q \in \{0,1,2,\ldots\}$
$$
	(1+2p)\alpha = q\mu'
	\quad \Rightarrow \quad 
	q\nu' = 1+2p,
$$
which requires $\nu'$ to be odd, so that $k=1$ 
(and $\nu$ is already odd if $\mu$ was even).
We also note that if any $\nu$-particle clusters happen to be enclosed in 
the 2-anyon exchange loops, they each
contribute $\nu$ to $p$, that is $2\nu\alpha$ to the magnetic phase,
and at the same time a matching $-2\mu$ to the relative momentum
(one $-\mu$ for each of the anyons orbiting around the cluster).
Furthermore, on a length scale such that a typical pair of particles has 
exactly one multiple of $\mu$ as orbital angular momentum, 
a full cancellation with the magnetic phase demands $p = (\nu-1)/2$ 
(this is indeed an integer 
in the even-numerator case),
i.e. about half of the particles of a cluster are enclosed.
On the very smallest scales, 
i.e. much less than the average interparticle distance,
we can accept a phase mismatch and strong 
repulsion, to be controlled by $\Phi$, in analogy with e.g. the 
hard-core Bose gas whose energy vanishes logarithmically with low density 
in two dimensions \cite{Schick-71,LieYng-01}.
In the case $\alpha = 2/3$ it is natural that as the scale then
increases a bit we first see exactly one enclosed anyon, $p=1$, 
of the $\nu=3$ cluster of which the particle pair is taken, 
and then additional full clusters on the scale of the average interparticle distance
(cf. \cite[Fig.~3]{LarLun-16} and Figure~\ref{fig:clusterplots-2p}). 
Hence this allows for a full cancellation of the exchange phase 
over large scales in this particular state,
thus reducing its overall statistical repulsion 
on large regions of the configuration space.
However, in general it seems that there needs to be a delicate balance between 
$\Phi$ and $\psi_\alpha$ in order to obtain such special probability distributions,
and this remains the least understood aspect of these trial states at the moment
(see also \cite{Lundholm-13}).

\section{Ideal anyons in a harmonic trap} \label{sec:oscillator}

At this point one might 
worry about actually computing (or at least bounding) the energy 
of the proposed trial states. Fortunately, however, 
it turns out that $\psi_\alpha$ given in \eqref{eq:trial-even} for even-numerator $\alpha$
is an \emph{exact} (but singular) eigenfunction of the harmonic oscillator
Hamiltonian $\hat{H}_N$ with
(see \cite{Chou-91b,Chou-91a,MurLawBhaDat-92,Myrheim-99})
\begin{equation} \label{eq:exact-energy}
	\hat{H}_N\psi_\alpha 
	= \omega(N + \deg \psi_\alpha) \,\psi_\alpha,
\end{equation}
where the degree of the state is, by \eqref{eq:trial-momentum},
$$
	\deg\psi_\alpha = -\alpha\binom{N}{2}-L = -\alpha\frac{\nu-1}{2}N.
$$
Being a singular eigenfunction means that the identity \eqref{eq:exact-energy}
holds wherever $\psi_\alpha$ is smooth, namely
outside of the fat diagonal of the configuration space 
\begin{equation} \label{eq:diagonals}
	\bDelta 
	:= \{ \sx = (\bx_1,\ldots,\bx_N) \in (\R^2)^N : \text{$\exists\ j \neq k$ s.t. $\bx_j = \bx_k$} \}.
\end{equation}
Since $\psi_\alpha$ is not a true eigenstate (normalizable and in the domain) 
of the operator $\hat{H}_N$, its formal energy being
$$
	\omega(N + \deg \psi_\alpha) = \omega\left(1-\alpha\frac{\nu-1}{2}\right)N 
	< \omega N
$$
is not a contradiction to \eqref{eq:CS-bound}.
Also, by contrast, note that for $\psi_\alpha$ with odd-numerator $\alpha$
$$
	\deg \psi_\alpha \sim \nu \frac{\sqrt{8}}{3} K^{3/2} = O(N^{3/2}),
$$
although in this case $\psi_\alpha$ does not have the structure of an 
exact eigenfunction since it involves a polynomial in both $z_j$ and $\bar{z}_j$.

Two possible choices of regularizing symmetric functions $\Phi$ mentioned in \cite{LunSol-13b}, 
giving rise to the expected pairwise short-scale behavior 
$\sim |z_{jk}|^\alpha$ in $\Psi$,
could be
\begin{equation} \label{eq:Phi-parameter}
	\Phi_{r_0} = \prod_{j<k}|z_{jk}|^{2\alpha}(r_0^2 + |z_{jk}|^2)^{-\alpha},
\end{equation}
with a parameter $r_0>0$ to be optimized over, 
or the parameter-free (but less smooth)
\begin{equation} \label{eq:Phi-neighbor}
	\Phi_0 = \prod_{j=1}^N \prod_{k=1}^{\nu-1} |z_{j\,\nn_k(j)}|^{\alpha},
\end{equation}
where $\nn_k(j)$ denotes the $k$th nearest neighbor of particle $j$
(among particles of the set $A \subseteq \{1,2,\ldots,N\}$ 
if instead writing $\nn_k(j;A)$).
However, as seen below, an ansatz closer to the Bijl--Jastrow form 
\begin{equation} \label{eq:Phi-Jastrow}
	\Phi_{\mathrm{BJ}} = \prod_{j<k} f(|z_{jk}|),
\end{equation}
or the Dyson form \cite{Dyson-57}
\begin{equation} \label{eq:Phi-Dyson}
	\Phi_{\mathrm{D}} = f(|z_{2\,1}|) f(|z_{3\,\nn_1(3;1,2)}|) \ldots f(|z_{N\,\nn_1(N;1,2,\ldots,N-1)}|),
\end{equation}
as used for 2D Bose gases with suitable two-particle correlations $f$ 
(see \cite{LieSeiSolYng-05}), could be better.

We note that taking $\Phi_0$ as in \eqref{eq:Phi-neighbor} as a regulator raises
the degree of the state $\Psi = \Phi_0\psi_\alpha$ for even numerators
to formally (if \eqref{eq:exact-energy} were still to hold) produce the energy
$$
	\omega(N + \deg \Psi) 
	= \left(1 + \alpha(\nu-1) - \alpha\frac{\nu-1}{2}\right) \omega N
	= \left(1 + \alpha\frac{\nu-1}{2}\right) \omega N,
$$
which by \eqref{eq:trial-momentum} exactly matches the lower bound \eqref{eq:CS-bound}.
We also note that for odd-numerator $\alpha$
and magic (i.e. shell-filling) numbers $K$,
$$
	\omega(N + \deg \Psi)
	\sim \omega \,\nu \,\frac{\sqrt{8}}{3}K^{3/2} 
	= \frac{\sqrt{8}}{3} \nu^{-1/2} \omega N^{3/2} 
	= \frac{\sqrt{8}}{3} \sqrt{\alpha_*} \,\omega N^{3/2},
$$
which matches
both the average-field approximation \eqref{eq:avg-field} 
for $\alpha=\alpha_*$ and,
for small $\alpha_*$ and up to the value of the numerical constant,
the improved rigorous lower bound \eqref{eq:rigorous-bounds} that was
recently obtained by means of a Lieb--Thirring inequality \cite{LarLun-16}
as a consequence of the statistical repulsion \eqref{eq:statistical-repulsion}.

Using the above observations, we then shift the problem of estimating the 
ground-state energy $E_0(N)$ from above for even-numerator fractions 
to concern only the energy of the regulator $\Phi$:

\begin{prop} \label{prop:Phi-bound-oscillator}
	Assume $\Phi \in H^1_{\loc}(\R^{2N};\R)$ 
	is such that $\Psi = \Phi\psi_\alpha \in \domD{N}{\alpha}$
	and $(\nabla\Phi)\psi_\alpha \in L^2(\R^{2N};\C^N)$,
	where $\alpha$ is an even-numerator fraction.
	Then
	\begin{align*}
		\langle\Psi, \hat{H}_N \Psi\rangle &= 
		\int_{\R^{2N}} \left( \frac{1}{2m}\sum_{j=1}^N |D_j \Psi|^2 + \frac{m\omega^2}{2} |\sx|^2 |\Psi|^2 \right) d\sx \\
		&= \omega\left( N + \deg \psi_\alpha \right) \int_{\R^{2N}} |\Psi|^2 \,d\sx
		+ \frac{1}{2m} \int_{\R^{2N}} \sum_{j=1}^N |\nabla_j \Phi|^2 |\psi_\alpha|^2 \,d\sx.
	\end{align*}
\end{prop}
\begin{proof}
	By taking an appropriate approximating sequence in 
	$H^1(\R^{2N}) = H^1_0(\R^{2N} \setminus \bDelta)$ 
	(see \cite[Lemma 3]{LunSol-14})
	we may assume without loss of generality that
	$\Phi \in C_c^\infty(\R^{2N} \setminus \bDelta)$
	(smooth with compact support outside $\bDelta$).
	Then 
	$$
		D_j \Psi = (-i\nabla_j \Phi) \psi_\alpha + \Phi D_j\psi_\alpha
		\ \ \in C_c^\infty(\R^{2N} \setminus \bDelta),
	$$
	and
	$$
		\sum_j \int |D_j \Psi|^2 \,d\sx
		= \sum_j \int \overline{D_j\Psi} \cdot (-i\nabla_j\Phi)\psi_\alpha \,d\sx
		+ \sum_j \int \overline{D_j\Psi} \cdot \Phi D_j\psi_\alpha \,d\sx,
	$$
	where, using that $\alpha\bA_j$ is real and a partial integration,
	\begin{align*}
		&\int \overline{D_j\Psi} \cdot \Phi D_j\psi_\alpha \,d\sx
		= \int i\nabla_j\overline{\Psi} \cdot \Phi D_j\psi_\alpha \,d\sx
		+ \int \alpha\bA_j\overline{\Psi} \cdot \Phi D_j\psi_\alpha \,d\sx \\
		&= \int \overline{\Psi} (-i\nabla_j \Phi) \cdot D_j\psi_\alpha \,d\sx
		+ \int \overline{\Psi} \Phi (-i\nabla_j \cdot D_j\psi_\alpha) \,d\sx
		+ \int \overline{\Psi}\Phi \alpha\bA_j \cdot D_j\psi_\alpha \,d\sx \\
		&= -\int \overline{(-i\nabla_j \Phi)\psi_\alpha} \cdot 
			\Big( D_j\Psi - (-i\nabla_j\Phi)\psi_\alpha \Big) \,d\sx
		+ \int \overline{\Psi}\Phi D_j \cdot D_j\psi_\alpha \,d\sx.
	\end{align*}
	Hence,
	\begin{align*}
		\int &\left( \frac{1}{2m} \sum_j |D_j \Psi|^2 + \frac{m\omega^2}{2} |\sx|^2 |\Psi|^2 \right) d\sx \\
		&= \int \overline{\Psi}\Phi \hat{H}_N \psi_\alpha \,d\sx
		+ \frac{1}{2m}\int \sum_j |\nabla_j \Phi|^2 |\psi_\alpha|^2 \,d\sx \\
		&\qquad
		+ \frac{1}{2m}\left( \int \sum_j \overline{D_j \Psi} \cdot (-i\nabla_j \Phi)\psi_\alpha \,d\sx
		- \int \sum_j \overline{(-i\nabla_j \Phi)\psi_\alpha} \cdot D_j \Psi \,d\sx \right),
	\end{align*}
	and by \eqref{eq:exact-energy} 
	it then remains to prove that the last line is zero.
	Expanding the derivative and collecting the terms, 
	and making another partial integration 
	(now for $\Phi^2 = |\Phi|^2 \in C_c^\infty(\R^{2N} \setminus \bDelta)$), 
	one finds that it equals
	$$
		\frac{i}{2m} \sum_j \int_{\R^{2N}} |\Phi|^2 \nabla_j \cdot (\bJ_j[\psi_\alpha] + \alpha\bA_j|\psi_\alpha|^2) \,d\sx,
	$$
	where $\bJ[u] := \frac{i}{2}(u\nabla\bar{u} - \bar{u}\nabla u)$.
	Finally we may use that $\nabla \cdot \bA = 0$ and the eigenfunction equation
	\eqref{eq:exact-energy} and its complex conjugate 
	to show that
	$$
		\sum_j \nabla_j \cdot (\bJ_j[\psi_\alpha] + \alpha\bA_j|\psi_\alpha|^2) = 0
	$$
	on $\R^{2N} \setminus \bDelta$, which proves the proposition.
\end{proof}

The energy in the state $\Psi$
thus depends solely on the correlations of the weight 
$|\psi_\alpha|^2$
and its balance with the regulator $\Phi$, 
which is required to vanish sufficiently fast as particles come together.
However, if these correlations effectively turn out to 
decay faster than the average interparticle spacing then,
in analogy with dilute hard-core bosons \cite{Schick-71,LieYng-01},
there may be room for a smaller energy
(note that $\Phi$ should not have Dirichlet but rather Neumann-type 
boundary conditions at the interparticle scale).

\section{The $R$-extended anyon gas} \label{sec:extended}

That exact eigenstates for the many-anyon problem can be found at all
is far from trivial, 
and the reason for it to hold for the above states 
is that they satisfy a remarkable simplifying identity.
Here we shall consider this identity in detail and greater generality,
in the context of the extended anyon gas.

By an `$R$-extended anyon' we mean that we have replaced the singular 
Aharonov--Bohm flux on each anyon by a uniform field on a 
disk of finite radius $R > 0$.
In other words, we replace \eqref{eq:anyon-potential-ideal} by
(cf. \cite{ChoLeeLee-92,Trugenberger-92b,Mashkevich-96,LunRou-15,LunRou-16,LarLun-16})
\begin{equation} \label{eq:anyon-potential-extended}
	\bA_j(\bx) := \sum_{k \neq j} \frac{(\bx-\bx_k)^\perp}{|\bx-\bx_k|_R^2},
	\qquad |\bx|_R := \max\{|\bx|,R\},
\end{equation}
so that, 
$$
	\curl \alpha\bA_j(\bx) 
	= 2\pi\alpha \sum_{k \neq j} \frac{\1_{D(\bx_k,R)}(\bx)}{\pi R^2}
	\ \xrightarrow{R \to 0} \ 
	2\pi\alpha \sum_{k \neq j} \delta_{\bx_k} (\bx),
$$
where $\1_{D(\by,R)}$ 
denotes the indicator function on a disk of 
radius $R$ centered at $\by$.
Note that this form for the magnetic interaction is actually the 
natural one from the perspective 
of emergent anyons \cite{LunRou-16}, for which the size $R$ is implied by the
experimental conditions.
There is also a natural dimensionless parameter in the problem
given by the ratio of the size of
the magnetic flux to the average interparticle distance, 
$\bar\gamma := R\bar\varrho^{1/2}$.
This 
has been
called the `magnetic filling ratio' in 
\cite{Trugenberger-92b,Trugenberger-92,LarLun-16}.

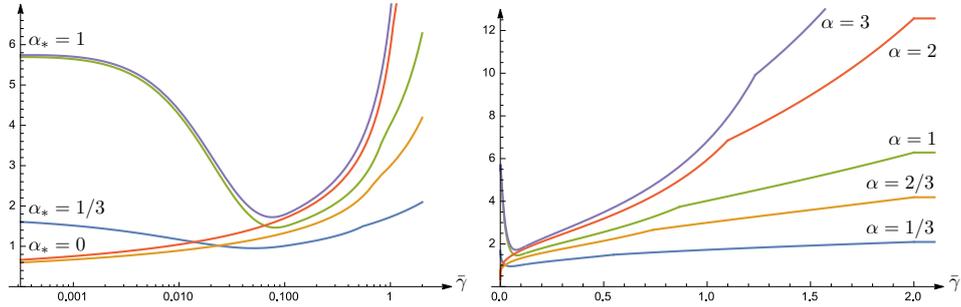
\begin{figure}[t]
	\centering
	\scalebox{0.84}{
	\begin{tikzpicture}
		\node [above right] at (0,0) {\includegraphics[scale=0.82, clip, trim=10pt 0pt 0pt 0pt]{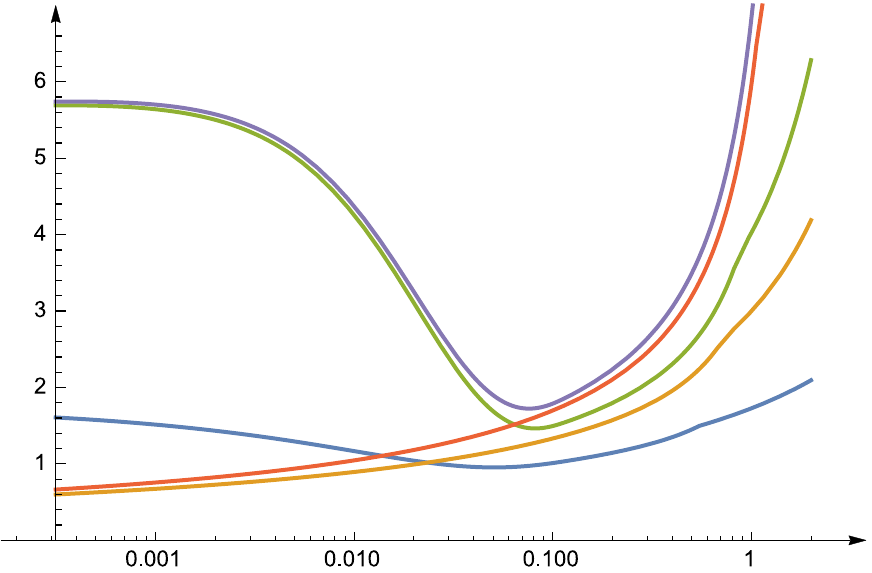}};
		\node [above right] at (7.0,.2) {\scalebox{0.8}{$\bar\gamma$\hspace{-5pt}}};
		\node [above right] at (0.3,0.8) {\scalebox{0.8}{$\alpha_*=0$\hspace{-5pt}}};
		\node [above right] at (0.3,1.35) {\scalebox{0.8}{$\alpha_*=1/3$\hspace{-5pt}}};
		\node [above right] at (0.3,4.04) {\scalebox{0.8}{$\alpha_*=1$\hspace{-5pt}}};
	\end{tikzpicture}
	\begin{tikzpicture}
		\node [above right] at (0,0) {\includegraphics[scale=0.82]{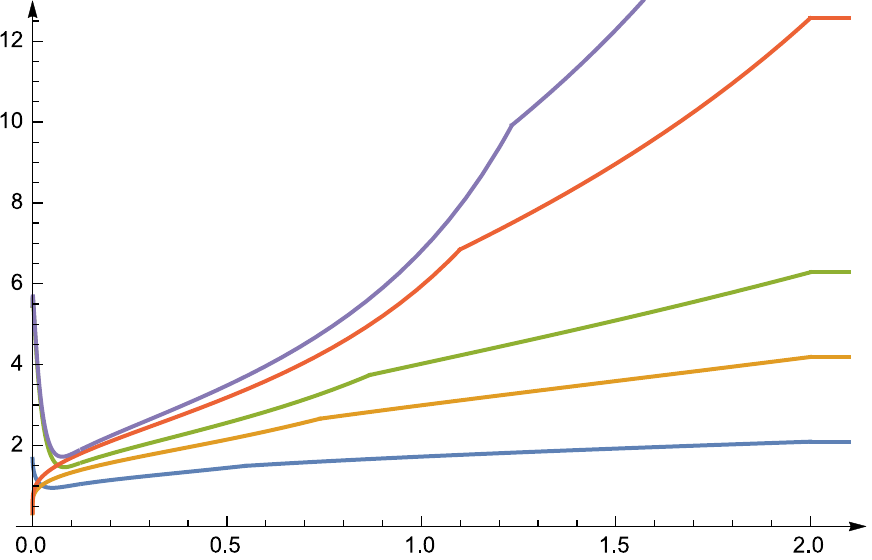}};
		\node [above right] at (7.3,.2) {\scalebox{0.8}{$\bar\gamma$\hspace{-5pt}}};
		\node [above right] at (6,1.04)   {\scalebox{0.8}{$\alpha=1/3$\hspace{-5pt}}};
		\node [above right] at (6,1.74) {\scalebox{0.8}{$\alpha=2/3$\hspace{-5pt}}};
		\node [above right] at (6.35,2.52) {\scalebox{0.8}{$\alpha=1$\hspace{-5pt}}};
		\node [above right] at (6.35,3.9) {\scalebox{0.8}{$\alpha=2$\hspace{-5pt}}};
		\node [above right] at (5.3,4.35) {\scalebox{0.8}{$\alpha=3$\hspace{-5pt}}};
	\end{tikzpicture}
	}
	\caption{The dependence on the density $\bar\gamma$
	of the lower bound 
	$e(\alpha, \bar\gamma)$ for the ground-state energy of the 
	extended anyon gas for some values of $\alpha$,
	with constants and scales chosen for illustrative purposes
	as explained in \cite{LarLun-16}.}
	\label{fig:EnergyVsGammaPlots}
\end{figure}

Based on the long-range local exclusion principle \eqref{eq:statistical-repulsion} 
and further short-range magnetic bounds that arise only 
in this extended context, 
it was shown rigorously in \cite{LarLun-16} that a
homogeneous gas of such $R$-extended anyons satisfies in the thermodynamic limit 
(on a box of side length $L \to \infty$ and with fixed average density 
$\bar\varrho = N/L^2$)
a universal bound for the energy per particle of the form\footnote{We
are taking the $\liminf$ (and assume Dirichlet boundary conditions)
here because, although the sequence is bounded,
it has not yet been proved in general that a limit exists.}
\begin{equation} \label{eq:ext-gas-bound}
	\liminf_{\substack{N, L \,\to\, \infty \\ N/L^2 = \bar\varrho}} \frac{E_0(N)}{N} 
	\ge C e(\alpha,\bar\gamma) \frac{\bar\varrho}{2m},
\end{equation}
where $C$ is a positive universal constant and 
(see Figure~\ref{fig:EnergyVsGammaPlots})
$$
	e(\alpha, \bar\gamma) \sim \left\{ \begin{array}{ll}
		\frac{2\pi}{|{\ln \bar\gamma}|} 
		+ \pi(j_{\alpha_*}')^2 \ge 2\pi \alpha_*, &\quad \bar\gamma \to 0 
			\ (\text{fixed} \ \alpha \neq 0), \\[6pt]
		2\pi|\alpha|, &\quad \bar\gamma \gtrsim 1. 
		\end{array}\right.
$$
This bound interpolates between a dilute regime where the effect of the
statistical repulsion dominates (note that there is also, even for 
$\alpha \in 2\Z \setminus \{0\}$, a strictly positive interaction 
energy which vanishes with the density similarly 
to that of a hard-core 2D Bose gas),
and a dense regime where the dependence on $\alpha$ matches that which is 
expected from average-field theory \eqref{eq:avg-field}.
In \cite{LunRou-15} it was shown that for $R$-extended anyons 
in an external trap $V$
and in a limit such that the filling $\bar\gamma$ 
is high but the statistics parameter $\alpha$ small 
(`almost-bosonic' anyons; see also \cite{CorLunRou-16}),
the average-field approximation is a correct description  
in the sense that the particles become identically distributed in a 
self-generated magnetic field.
However, given the \emph{linear} 
dependence on the strength of the magnetic field $|\alpha|$
(our lower bound \eqref{eq:ext-gas-bound} 
is valid for any $\alpha \in \R$) for high densities, 
and the \emph{periodicity} in $\alpha$ for ideal anyons,
there must be some non-trivial interpolation between these two regimes.

Let us now introduce a convenient notation for an associated scalar
(super)potential (cf. \cite{LunRou-15})
$$
	w_R(\bx) := \left\{ \begin{array}{ll}
		\ln |\bx|, & |\bx| > R, \\
		\ln R + \frac{1}{2}\left(|\bx|^2/R^2 - 1\right), & |\bx| \le R,
		\end{array} \right.
$$
with
$$
	\nabla w_R(\bx) = (\bx)_R^{-1} := \bx/|\bx|_R^2, 
$$
and
$$
	W_R(\sx) := \sum_{\substack{j,k=1 \\ j \neq k}}^N \Delta w_R(\bx_j-\bx_k)
	= 2\pi\sum_{\substack{j,k=1 \\ j \neq k}}^N 
		\frac{\1_{D(0,R)}}{\pi R^2}(\bx_j-\bx_k).
$$
We then have the following property, 
which is essentially a result concerning 
supersymmetry of the corresponding Pauli operator.
It has been discussed in that context in \cite{Girvin-etal-90,ChoLeeLee-92},
however we will here supply a different proof.

\begin{prop} \label{prop:Pauli-identity}
	Let $\Psi_{\pm}(\sx) = e^{\mp \alpha\sum_{j<k} w_R(\bx_j-\bx_k)} f_\mp(\sz)$,
	where $f_+$ is analytic resp. $f_-$ anti-analytic in all the variables $z_j$.
	Then
	$$
		\sum_{j=1}^N D_j^2 \,\Psi_\pm 
		= \pm \alpha W_R \,\Psi_\pm.
	$$
	In particular, $\Psi_{\pm}$ are for $R=0$ generalized 
	zero-energy eigenfunctions of the $N$-anyon kinetic energy operator 
	$\hat{T}_\alpha$ considered on $\R^{2N} \setminus \bDelta$.
\end{prop}
\begin{proof}
	We find it convenient to work with $\cG(\C^2)$, the complex
	Clifford algebra over $\R^2$, and write for the $\pi/2$-rotation
	$\bx^\perp = (x\be_1+y\be_2)^\perp = \bx I$, 
	i.e. multiplication from the right with the pseudoscalar $I = \be_1\be_2$.
	Note that
	$$
		\nabla_j \Psi_\pm = e^{\mp \alpha\sum_{i<k} w_R(\bx_i-\bx_k)}
		\left( \mp \alpha \sum_{k \neq j}(\bx_j-\bx_k)_R^{-1} f_\mp 
			+ \nabla_j f_\mp \right),
	$$
	and thus
	$$
		D_j \Psi_\pm = e^{\mp \alpha\sum_{i<k} w_R(\bx_i-\bx_k)}
		\left( \alpha \sum_{k \neq j}(\bx_j-\bx_k)_R^{-1} (\pm i + I) f_\mp
			-i\nabla_j f_\mp \right).
	$$
	Furthermore,
	\begin{align*}
		D_j \cdot D_j \Psi_\pm 
		&= \left( \left(D_j e^{\mp \alpha\sum_{i<k} w_R(\bx_i-\bx_k)}\right) 
			-i e^{\mp \alpha\sum_{i<k} w_R(\bx_i-\bx_k)} \nabla_j \right) \\
		&\qquad\cdot \left( \pm i\alpha\sum_{k \neq j}(\bx_j-\bx_k)_R^{-1} (1 \mp iI)f_\mp 
			-i\nabla_j f_\mp \right) \\
		&= \pm \alpha \sum_{k \neq j} \nabla_j \cdot (\bx_j - \bx_k)_R^{-1}(1 \mp iI) \Psi_\pm \\
		&= \pm \alpha \sum_{k \neq j} \Delta w_R(\bx_j - \bx_k) \Psi_\pm,
	\end{align*}
	where the fundamental simplifying identity used 
	is that for any
	$\bx,\by \in \C^2$
	$$
		\bx(1\pm iI) \cdot \by(1\pm iI) = 0,
	$$
	since $I^2 = -1$ and $\bx I \cdot \by I = \bx \cdot \by$.
	We have also used $\nabla \cdot \nabla I = \nabla \cdot \nabla^\perp = 0$,
	and that
	for $z_\pm = x \pm iy$ and $f\colon \C \to \C$ analytic 
	$$
		\nabla(f(z_\pm)) = f'(z_\pm) (\be_1 \pm i\be_2) = f'(z_\pm) \be_1(1\pm iI),
	$$
	and $\Delta f = 0$.
\end{proof}

In the $R$-extended case we therefore take as our trial states 
$\Psi = \Phi \psi_\alpha$ with 
the Jastrow factor in \eqref{eq:trial-even} and \eqref{eq:trial-odd} 
replaced by $e^{-\alpha\sum_{j<k} w_R(\bx_j-\bx_k)}$,
and in the case of the homogeneous gas the $\varphi_k$
are taken to be the eigenstates of the Neumann Laplacian on the square $Q_L$
of side length $L$ (thus $\varphi_0 \equiv L^{-1}$).
Note that these states are regular even without the factor $\Phi$, since
$$
	e^{-\alpha w_R(\bx)} = \left\{ \begin{array}{ll}
		|\bx|^{-\alpha}, & |\bx| > R, \\
		R^{-\alpha} e^{\frac{\alpha}{2}(1-|\bx|^2/R^2)}, & |\bx| \le R.
		\end{array} \right.
$$
However, in order to obtain the correct balance for a low total energy, 
and to take the appropriate limits, we expect that an additional regulator
is still necessary. In particular, in the dilute limit 
$\bar\gamma = R\bar\varrho^{1/2} \to 0$ the Jastrow factor describes an 
attraction which needs to be turned into a short-range repulsion,
as illustrated by the below 
reformulation of the energy in terms of $\Phi$.
	The proof is almost identical to that of Proposition~\ref{prop:Phi-bound-oscillator},
	where the use of the identity \eqref{eq:exact-energy} 
	is replaced by Proposition~\ref{prop:Pauli-identity}.

\begin{prop} \label{prop:Phi-bound-extended}
	Assume $\Phi \in H^1_0(Q_L^N;\R)$ 
	is such that $\Psi = \Phi\psi_\alpha \in \domD{N}{\alpha}$ 
	and $(\nabla\Phi)\psi_\alpha \in L^2(Q_L^N)$,
	where $\alpha \in [0,1]$ is an even-numerator fraction.
	Then
	$$
		\int_{Q_L^N} \sum_{j=1}^N |D_j \Psi|^2 \,d\sx 
		= \int_{Q_L^N} \left( 
			\sum_{j=1}^N |\nabla_j \Phi|^2 + \alpha W_R |\Phi|^2
			\right) |\psi_\alpha|^2 \,d\sx.
	$$
\end{prop}

In the dilute limit, in which the scattering length of the soft-disk
potential $\alpha W_R$ becomes relatively small,
this again seems to be able to produce a low energy for even-numerator states.
Also, for odd-numerator states, naively estimating the energy of $\Psi$ in terms 
of that of the one-body states $\varphi_k$ of the Slater determinants in
\eqref{eq:trial-odd} yields the tentative bound
$$
	2m E_0(N) \ \lesssim \ \nu\, 2\pi K^2/L^2 
	\ \sim \ 2\pi \alpha_* \bar\varrho \,N,
$$
which again matches the available lower bounds.
Note also that the repulsive pair potential $\alpha W_R$ 
that \emph{emerged} above matches in the dilute limit the 
point interaction conventionally introduced to regularize
ideal anyons \cite{Ouvry-94}.

\section{Conclusions} \label{sec:conclusions}

With the ansatz given by the discussed trial states,
we have reduced the difficult problem of bounding 
the ground-state energy of a system of $N$ abelian
anyons with even-numerator rational statistics parameter
to the study of the $N$-dependence of the quantity
$$
	\frac{ \int_{\R^{2N}} \left(
			\sum_{j=1}^N |\nabla_j \Phi|^2 + \alpha W_R |\Phi|^2
			\right) |\psi_\alpha|^2 \,d\sx 
		}{ \int_{\R^{2N}} |\Phi|^2 |\psi_\alpha|^2 \,d\sx },
$$
which is essentially the energy of a repulsive 2D Bose gas described by $\Phi$ 
but weighted by $|\psi_\alpha|^2$.
One could try to estimate this using the techniques of Dingle, Jastrow and Dyson
(see \cite{Dyson-57,LieSeiSolYng-05} and references therein).
Alternatively, Monte Carlo methods could prove useful in this formulation.
In any case, since the weight $|\psi_\alpha|^2$ is designed so as to cancel
any long-range correlations by means of its clustering properties, and
since the energy of a dilute 2D Bose gas is logarithmically
small \cite{Schick-71,LieYng-01}, the discussed approach indeed looks 
very promising.
Also, if the anyons are not completely free but an additional attraction
is added then it seems rather clear 
from the above expression with suitable 
$\Phi$ that they would prefer to cluster in this way.

Finally, let us remark that if these are indeed the correct (approximative)
ground states for a many-body system of abelian anyons, then they could possibly 
also explain from a more fundamental perspective
the occurrence of such clustering states in the FQHE 
(cf. \cite[pp.~239--240]{Stern-08} and note that the usual Read--Rezayi states are
supposed to be built of $k$-clusters of $\alpha = 2/k$ anyons 
in a zero magnetic field).
Furthermore, the elementary excitations of such an abelian anyon condensed ground state
may, according to well-known properties of clustering states, in turn be 
\emph{non}-abelian anyons.

\begin{ack}
	The idea for using 
	the clustering states \eqref{eq:trial-even} and \eqref{eq:trial-odd}
	for the many-anyon
	problem came during a postdoctoral stay at IH\'ES in the fall 2011
	as part of an EPDI fellowship, and, upon coming across the paper \cite{CapGeoTod-01}
	the following year I found out about their intriguing connection to 
	the Moore--Read and Read--Rezayi states of the FQHE.
	I am especially grateful to Jan Philip Solovej for fruitful discussions and
	collaboration on this topic, and for initiating our mathematical study of 
	anyons in the first place during my postdoc in Copenhagen.
	The plots of 
	Figures~\ref{fig:clusterplots-1p}-\ref{fig:EnergyVsGammaPlots}
	were produced in collaboration with Simon Larson.
	I also thank Michele Correggi, 
	Phan Th\`anh Nam, Fabian Portmann and
	Nicolas Rougerie for valuable discussions and collaborations on closely 
	related subjects, as well as 
	Eddy Ardonne, Hans Hansson, Thierry Jolicoeur, St\'ephane Ouvry,
	Raoul Santachiara, Robert Seiringer and Andrea Trombettoni
	for useful references, comments and discussions.
	Financial support from the Knut and Alice Wallenberg Foundation,
	grant no. KAW 2010.0063,
	and the Swedish Research Council, grant no. 2013-4734,
	is gratefully acknowledged.
\end{ack}


\def\MR#1{} 

\begin{thebibliography}{10}

\bibitem{LeiMyr-77}
J.~M. {Leinaas} and J. {Myrheim}, \emph{{On the theory of identical
  particles}}, Nuovo Cimento B \textbf{37} (1977), 1--23,
  \href{http://dx.doi.org/10.1007/BF02727953}{\path{doi}}.

\bibitem{GolMenSha-81}
G.~A. Goldin, R. Menikoff, and D.~H. Sharp, \emph{Representations of a local
  current algebra in nonsimply connected space and the {A}haronov-{B}ohm
  effect}, J. Math. Phys. \textbf{22} (1981), no.~8, 1664--1668,
  \href{http://dx.doi.org/10.1063/1.525110}{\path{doi}}.

\bibitem{Wilczek-82a}
F. Wilczek, \emph{Magnetic flux, angular momentum, and statistics}, Phys. Rev.
  Lett. \textbf{48} (1982), 1144--1146,
  \href{http://dx.doi.org/10.1103/PhysRevLett.48.1144}{\path{doi}}.

\bibitem{Wilczek-82b}
F. Wilczek, \emph{Quantum mechanics of fractional-spin particles}, Phys. Rev.
  Lett. \textbf{49} (1982), 957--959,
  \href{http://dx.doi.org/10.1103/PhysRevLett.49.957}{\path{doi}}.

\bibitem{Wu-84}
Y.-S. Wu, \emph{Multiparticle quantum mechanics obeying fractional statistics},
  Phys. Rev. Lett. \textbf{53} (1984), 111--114,
  \href{http://dx.doi.org/10.1103/PhysRevLett.53.111}{\path{doi}}.

\bibitem{Souriau-70}
J.-M. Souriau, \emph{Structure des syst\`emes dynamiques}, Ma{\^{i}}trises de
  math\'ematiques, Dunod, Paris, 1970, English translation by R. H. Cushman and
  G. M. Tuynman, Progress in Mathematics, 149, Birkh\"auser Boston Inc.,
  Boston, MA, 1997,
  \url{http://www.jmsouriau.com/structure_des_systemes_dynamiques.htm}.
  \MR{0260238}

\bibitem{DatMurVat-03}
G. Date, M.~V.~N. Murthy, and R. Vathsan, \emph{Classical and quantum mechanics
  of anyons}, arXiv e-prints, 2003,
  \href{http://arxiv.org/abs/cond-mat/0302019}{\path{arXiv:cond-mat/0302019}}.

\bibitem{Forte-92}
S. Forte, \emph{Quantum mechanics and field theory with fractional spin and
  statistics}, Rev. Mod. Phys. \textbf{64} (1992), 193--236,
  \href{http://dx.doi.org/10.1103/RevModPhys.64.193}{\path{doi}}.

\bibitem{Froehlich-90}
J. Fr{\"o}hlich, \emph{Quantum statistics and locality}, Proceedings of the
  {G}ibbs {S}ymposium ({N}ew {H}aven, {CT}, 1989), Amer. Math. Soc.,
  Providence, RI, 1990, pp.~89--142. \MR{1095329}

\bibitem{IenLec-92}
R. Iengo and K. Lechner, \emph{Anyon quantum mechanics and {C}hern-{S}imons
  theory}, Phys. Rep. \textbf{213} (1992), 179--269,
  \href{http://dx.doi.org/10.1016/0370-1573(92)90039-3}{\path{doi}}.

\bibitem{Khare-05}
A. Khare, \emph{{Fractional Statistics and Quantum Theory}}, 2nd ed., World
  Scientific, Singapore, 2005.

\bibitem{Lerda-92}
A. Lerda, \emph{{Anyons}}, Springer-Verlag, Berlin--Heidelberg, 1992.

\bibitem{Myrheim-99}
J. Myrheim, \emph{Anyons}, Topological aspects of low dimensional systems (A.
  Comtet, T. Jolic{\oe}ur, S. Ouvry, and F. David, eds.), Les Houches - Ecole
  d'Ete de Physique Theorique, vol.~69, (Springer-Verlag, Berlin, Germany),
  1999, pp.~265--413,
  \href{http://dx.doi.org/10.1007/3-540-46637-1_4}{\path{doi}}.

\bibitem{Ouvry-07}
S. Ouvry, \emph{{Anyons and lowest Landau level anyons}}, S\'eminaire
  Poincar\'e \textbf{11} (2007), 77--107,
  \href{http://dx.doi.org/10.1007/978-3-7643-8799-0_3}{\path{doi}}.

\bibitem{Stern-08}
A. Stern, \emph{{Anyons and the quantum Hall effect -- A pedagogical review}},
  Ann. Phys. \textbf{323} (2008), no.~1, 204--249, January Special Issue 2008,
  \href{http://dx.doi.org/10.1016/j.aop.2007.10.008}{\path{doi}}.

\bibitem{Wilczek-90}
F. Wilczek, \emph{{Fractional Statistics and Anyon Superconductivity}}, World
  Scientific, Singapore, 1990.

\bibitem{Girvin-04}
S. Girvin, \emph{Introduction to the fractional quantum {H}all effect},
  S\'eminaire Poincar\'e \textbf{2} (2004), 54--74,
  \href{http://dx.doi.org/10.1007/3-7643-7393-8_4}{\path{doi}}.

\bibitem{Goerbig-09}
M.~O. Goerbig, \emph{Quantum {H}all effects}, Lecture notes, 2009,
  \href{http://arxiv.org/abs/0909.1998}{\path{arXiv:0909.1998}}.

\bibitem{Jain-07}
J.~K. Jain, \emph{{Composite fermions}}, Cambridge Univ. Press, 2007.

\bibitem{Laughlin-99}
R.~B. Laughlin, \emph{Nobel lecture: Fractional quantization}, Rev. Mod. Phys.
  \textbf{71} (1999), 863--874,
  \href{http://dx.doi.org/10.1103/RevModPhys.71.863}{\path{doi}}.

\bibitem{StoTsuGos-99}
H. St\"{o}rmer, D. Tsui, and A. Gossard, \emph{The fractional quantum {H}all
  effect}, Rev. Mod. Phys. \textbf{71} (1999), S298--S305,
  \href{http://dx.doi.org/10.1103/RevModPhys.71.S298}{\path{doi}}.

\bibitem{BloDalZwe-08}
I. Bloch, J. Dalibard, and W. Zwerger, \emph{Many-body physics with ultracold
  gases}, Rev. Mod. Phys. \textbf{80} (2008), no.~3, 885--964,
  \href{http://dx.doi.org/10.1103/RevModPhys.80.885}{\path{doi}}.

\bibitem{Cooper-08}
N.~R. {Cooper}, \emph{{Rapidly rotating atomic gases}}, Advances in Physics
  \textbf{57} (2008), no.~6, 539--616,
  \href{http://dx.doi.org/10.1080/00018730802564122}{\path{doi}}.

\bibitem{MorFed-07}
A. Morris and D. Feder, \emph{Gaussian potentials facilitate access to quantum
  {H}all states in rotating {B}ose gases}, Phys. Rev. Lett. \textbf{99} (2007),
  240401, \href{http://dx.doi.org/10.1103/PhysRevLett.99.240401}{\path{doi}}.

\bibitem{RonRizDal-11}
M. Roncaglia, M. Rizzi, and J. Dalibard, \emph{From rotating atomic rings to
  quantum {H}all states}, www.nature.com, Scientific Reports \textbf{1} (2011),
  43, \href{http://dx.doi.org/10.1038/srep00043}{\path{doi}}.

\bibitem{Viefers-08}
S. Viefers, \emph{Quantum {H}all physics in rotating {B}ose-{E}instein
  condensates}, J. Phys. C \textbf{12} (2008), 123202,
  \href{http://dx.doi.org/10.1088/0953-8984/20/12/123202}{\path{doi}}.

\bibitem{Bolotin-etal-09}
K.~I. Bolotin, F. Ghahari, M.~D. Shulman, H.~L. Stormer, and P. Kim,
  \emph{Observation of the fractional quantum {H}all effect in graphene},
  Nature \textbf{462} (2009), 196--199,
  \href{http://dx.doi.org/10.1038/nature08582}{\path{doi}}.

\bibitem{Du-etal-09}
X. Du, I. Skachko, F. Duerr, A. Luican, and E.~Y. Andrei, \emph{Fractional
  quantum {H}all effect and insulating phase of {D}irac electrons in graphene},
  Nature \textbf{462} (2009), 192--195,
  \href{http://dx.doi.org/10.1038/nature08522}{\path{doi}}.

\bibitem{AroSchWil-84}
D. Arovas, J.~R. Schrieffer, and F. Wilczek, \emph{Fractional statistics and
  the quantum {H}all effect}, Phys. Rev. Lett. \textbf{53} (1984), 722--723,
  \href{http://dx.doi.org/10.1103/PhysRevLett.53.722}{\path{doi}}.

\bibitem{LunRou-16}
D. Lundholm and N. Rougerie, \emph{Emergence of fractional statistics for
  tracer particles in a {L}aughlin liquid}, Phys. Rev. Lett. \textbf{116}
  (2016), 170401,
  \href{http://dx.doi.org/10.1103/PhysRevLett.116.170401}{\path{doi}}.

\bibitem{Rougerie-16}
N. Rougerie, \emph{Some contributions to many-body quantum mathematics},
  Habilitation thesis, 2016,
  \href{http://arxiv.org/abs/1607.03833}{\path{arXiv:1607.03833}}.

\bibitem{AroSchWilZee-85}
D.~P. Arovas, R. Schrieffer, F. Wilczek, and A. Zee, \emph{Statistical
  mechanics of anyons}, Nuclear Physics B \textbf{251} (1985), 117 -- 126,
  \href{http://dx.doi.org/10.1016/0550-3213(85)90252-4}{\path{doi}}.

\bibitem{SpoVerZah-91}
M. Sporre, J.~J.~M. Verbaarschot, and I. Zahed, \emph{Numerical solution of the
  three-anyon problem}, Phys. Rev. Lett. \textbf{67} (1991), 1813--1816,
  \href{http://dx.doi.org/10.1103/PhysRevLett.67.1813}{\path{doi}}.

\bibitem{MurLawBraBha-91}
M.~V.~N. Murthy, J. Law, M. Brack, and R.~K. Bhaduri, \emph{Quantum spectrum of
  three anyons in an oscillator potential}, Phys. Rev. Lett. \textbf{67}
  (1991), 1817--1820,
  \href{http://dx.doi.org/10.1103/PhysRevLett.67.1817}{\path{doi}}.

\bibitem{SpoVerZah-92}
M. Sporre, J.~J.~M. Verbaarschot, and I. Zahed, \emph{Four anyons in a harmonic
  well}, Phys. Rev. B \textbf{46} (1992), 5738--5741,
  \href{http://dx.doi.org/10.1103/PhysRevB.46.5738}{\path{doi}}.

\bibitem{SpoVerZah-93}
M. Sporre, J. Verbaarschot, and I. Zahed, \emph{Anyon spectra and the third
  virial coefficient}, Nuclear Physics B \textbf{389} (1993), no.~3, 645--665,
  \href{http://dx.doi.org/10.1016/0550-3213(93)90357-U}{\path{doi}}.

\bibitem{ChenWilWitHal-89}
Y.~H. Chen, F. Wilczek, E. Witten, and B.~I. Halperin, \emph{On anyon
  superconductivity}, Int. J. Mod. Phys. B \textbf{3} (1989), 1001--1067,
  \href{http://dx.doi.org/10.1142/S0217979289000725}{\path{doi}}.

\bibitem{Hosotani-93}
Y. Hosotani, \emph{Neutral and charged anyon fluids}, Int. J. Mod. Phys. B
  \textbf{7} (1993), 2219,
  \href{http://arxiv.org/abs/cond-mat/9302002}{\path{arXiv:cond-mat/9302002}},
  \href{http://dx.doi.org/10.1142/S0217979293002857}{\path{doi}}.

\bibitem{WenZee-90}
X.~G. Wen and A. Zee, \emph{Compressibility and superfluidity in the
  fractional-statistics liquid}, Phys. Rev. B \textbf{41} (1990), 240--253,
  \href{http://dx.doi.org/10.1103/PhysRevB.41.240}{\path{doi}}.

\bibitem{DasOuv-94}
A. {Dasni{\`e}res de Veigy} and S. Ouvry, \emph{Equation of state of an anyon
  gas in a strong magnetic field}, Phys. Rev. Lett. \textbf{72} (1994), 600,
  \href{http://dx.doi.org/10.1103/PhysRevLett.72.600}{\path{doi}}.

\bibitem{Minor-93}
W.~R. Minor, \emph{Ground-state energy of a dilute anyon gas}, Phys. Rev. B
  \textbf{47} (1993), 12716--12721,
  \href{http://dx.doi.org/10.1103/PhysRevB.47.12716}{\path{doi}}.

\bibitem{LunSol-13a}
D. Lundholm and J.~P. Solovej, \emph{{Hardy and Lieb-Thirring inequalities for
  anyons}}, Comm. Math. Phys. \textbf{322} (2013), 883--908,
  \href{http://dx.doi.org/10.1007/s00220-013-1748-4}{\path{doi}}.

\bibitem{LunSol-13b}
D. Lundholm and J.~P. Solovej, \emph{Local exclusion principle for identical
  particles obeying intermediate and fractional statistics}, Phys. Rev. A
  \textbf{88} (2013), 062106,
  \href{http://dx.doi.org/10.1103/PhysRevA.88.062106}{\path{doi}}.

\bibitem{LunSol-14}
D. Lundholm and J.~P. Solovej, \emph{{Local exclusion and Lieb-Thirring
  inequalities for intermediate and fractional statistics}}, Ann. Henri
  Poincar\'e \textbf{15} (2014), 1061--1107,
  \href{http://dx.doi.org/10.1007/s00023-013-0273-5}{\path{doi}}.

\bibitem{LunRou-15}
D. Lundholm and N. Rougerie, \emph{{The average field approximation for almost
  bosonic extended anyons}}, J. Stat. Phys. \textbf{161} (2015), no.~5,
  1236--1267, \href{http://dx.doi.org/10.1007/s10955-015-1382-y}{\path{doi}}.

\bibitem{LarLun-16}
S. Larson and D. Lundholm, \emph{Exclusion bounds for extended anyons}, arXiv
  e-prints, 2016,
  \href{http://arxiv.org/abs/1608.04684}{\path{arXiv:1608.04684}}.

\bibitem{CorLunRou-16}
M. Correggi, D. Lundholm, and N. Rougerie, \emph{Local density approximation
  for the almost-bosonic anyon gas}, arXiv e-prints, 2016,
  \href{http://arxiv.org/abs/1611.00942}{\path{arXiv:1611.00942}}.

\bibitem{Gentile-40}
G. Gentile, \emph{Osservazioni sopra le statistiche intermedie}, Il Nuovo
  Cimento \textbf{17} (1940), no.~10, 493--497,
  \href{http://dx.doi.org/10.1007/BF02960187}{\path{doi}}.

\bibitem{Gentile-42}
G. Gentile, \emph{Le statistiche intermedie e le propriet{\`a} dell'elio
  liquido}, Il Nuovo Cimento \textbf{19} (1942), no.~4, 109--125,
  \href{http://dx.doi.org/10.1007/BF02960192}{\path{doi}}.

\bibitem{ChiSen-92}
R. Chitra and D. Sen, \emph{{Ground state of many anyons in a harmonic
  potential}}, Phys. Rev. B \textbf{46} (1992), 10923--10930,
  \href{http://dx.doi.org/10.1103/PhysRevB.46.10923}{\path{doi}}.

\bibitem{Sen-92}
D. Sen, \emph{Some supersymmetric features in the spectrum of anyons in a
  harmonic potential}, Phys. Rev. D \textbf{46} (1992), 1846--1857,
  \href{http://dx.doi.org/10.1103/PhysRevD.46.1846}{\path{doi}}.

\bibitem{Trugenberger-92b}
C. Trugenberger, \emph{{Ground state and collective excitations of extended
  anyons}}, Phys. Lett. B \textbf{288} (1992), 121--128,
  \href{http://dx.doi.org/10.1016/0370-2693(92)91965-C}{\path{doi}}.

\bibitem{Haldane-91}
F.~D.~M. Haldane, \emph{{``Fractional statistics'' in arbitrary dimensions: A
  generalization of the Pauli principle}}, Phys. Rev. Lett. \textbf{67} (1991),
  937--940, \href{http://dx.doi.org/10.1103/PhysRevLett.67.937}{\path{doi}}.

\bibitem{Isakov-94}
S.~B. Isakov, \emph{Statistical mechanics for a class of quantum statistics},
  Phys. Rev. Lett. \textbf{73} (1994), 2150--2153,
  \href{http://dx.doi.org/10.1103/PhysRevLett.73.2150}{\path{doi}}.

\bibitem{Wu-94}
Y.-S. Wu, \emph{Statistical distribution for generalized ideal gas of
  fractional-statistics particles}, Phys. Rev. Lett. \textbf{73} (1994), 922,
  \href{http://dx.doi.org/10.1103/PhysRevLett.73.922}{\path{doi}}.

\bibitem{HanLeiVie-01}
T.~H. Hansson, J.~M. Leinaas, and S. Viefers, \emph{Exclusion statistics in a
  trapped two-dimensional {B}ose gas}, Phys. Rev. Lett. \textbf{86} (2001),
  2930--2933, \href{http://dx.doi.org/10.1103/PhysRevLett.86.2930}{\path{doi}}.

\bibitem{CanJoh-94}
G.~S. Canright and M.~D. Johnson, \emph{Fractional statistics: alpha to beta},
  J. Phys. A: Math. Gen. \textbf{27} (1994), no.~11, 3579,
  \href{http://dx.doi.org/10.1088/0305-4470/27/11/009}{\path{doi}}.

\bibitem{BorSor-92}
M. Bourdeau and R.~D. Sorkin, \emph{When can identical particles collide?},
  Phys. Rev. D \textbf{45} (1992), 687--696,
  \href{http://dx.doi.org/10.1103/PhysRevD.45.687}{\path{doi}}.

\bibitem{LeeYan-57}
T.~D. Lee and C.~N. Yang, \emph{Many-body problem in quantum mechanics and
  quantum statistical mechanics}, Phys. Rev. \textbf{105} (1957), 1119--1120,
  \href{http://dx.doi.org/10.1103/PhysRev.105.1119}{\path{doi}}.

\bibitem{Dyson-57}
F.~J. Dyson, \emph{Ground-state energy of a hard-sphere gas}, Phys. Rev.
  \textbf{106} (1957), no.~1, 20--26,
  \href{http://dx.doi.org/10.1103/PhysRev.106.20}{\path{doi}}.

\bibitem{Schick-71}
M. Schick, \emph{Two-dimensional system of hard-core bosons}, Phys. Rev. A
  \textbf{3} (1971), 1067--1073,
  \href{http://dx.doi.org/10.1103/PhysRevA.3.1067}{\path{doi}}.

\bibitem{LieYng-98}
E.~H. Lieb and J. Yngvason, \emph{Ground state energy of the low density {B}ose
  gas}, Phys. Rev. Lett. \textbf{80} (1998), no.~12, 2504--2507,
  \href{http://dx.doi.org/10.1103/PhysRevLett.80.2504}{\path{doi}}.

\bibitem{LieYng-01}
E.~H. Lieb and J. Yngvason, \emph{The ground state energy of a dilute
  two-dimensional {B}ose gas}, J. Statist. Phys. \textbf{103} (2001), no.~3-4,
  509--526, Special issue dedicated to the memory of Joaquin M. Luttinger,
  \href{http://dx.doi.org/10.1023/A:1010337215241}{\path{doi}}. \MR{1827922}

\bibitem{LieSeiSolYng-05}
E.~H. Lieb, R. Seiringer, J.~P. Solovej, and J. Yngvason, \emph{The mathematics
  of the {B}ose gas and its condensation}, Oberwolfach {S}eminars,
  Birkh{\"a}user, 2005,
  \href{http://arxiv.org/abs/cond-mat/0610117}{\path{arXiv:cond-mat/0610117}}.

\bibitem{tHooft-88}
G. 't~Hooft, \emph{Non-perturbative 2 particle scattering amplitudes in 2+1
  dimensional quantum gravity}, Comm. Math. Phys. \textbf{117} (1988), no.~4,
  685--700, \href{http://dx.doi.org/10.1007/BF01218392}{\path{doi}}.

\bibitem{DesJac-88}
S. Deser and R. Jackiw, \emph{Classical and quantum scattering on a cone},
  Comm. Math. Phys. \textbf{118} (1988), no.~3, 495--509,
  \href{http://dx.doi.org/10.1007/BF01466729}{\path{doi}}.

\bibitem{KayStu-91}
B.~S. Kay and U.~M. Studer, \emph{Boundary conditions for quantum mechanics on
  cones and fields around cosmic strings}, Comm. Math. Phys. \textbf{139}
  (1991), no.~1, 103--139,
  \href{http://dx.doi.org/10.1007/BF02102731}{\path{doi}}.

\bibitem{LieThi-75}
E.~H. Lieb and W.~E. Thirring, \emph{Bound for the kinetic energy of fermions
  which proves the stability of matter}, Phys. Rev. Lett. \textbf{35} (1975),
  687--689, \href{http://dx.doi.org/10.1103/PhysRevLett.35.687}{\path{doi}}.

\bibitem{RegGoeJol-08}
N. Regnault, M.~O. Goerbig, and T. Jolicoeur, \emph{Bridge between abelian and
  non-abelian fractional quantum {H}all states}, Phys. Rev. Lett. \textbf{101}
  (2008), 066803,
  \href{http://dx.doi.org/10.1103/PhysRevLett.101.066803}{\path{doi}}.

\bibitem{ReaRez-99}
N. Read and E. Rezayi, \emph{{Beyond paired quantum Hall states: Parafermions
  and incompressible states in the first excited Landau level}}, Phys. Rev. B
  \textbf{59} (1999), 8084--8092,
  \href{http://dx.doi.org/10.1103/PhysRevB.59.8084}{\path{doi}}.

\bibitem{CapGeoTod-01}
A. Cappelli, L.~S. Georgiev, and I.~T. Todorov, \emph{Parafermion {H}all states
  from coset projections of abelian conformal theories}, Nucl. Phys. B
  \textbf{599} (2001), no.~3, 499 -- 530,
  \href{http://dx.doi.org/10.1016/S0550-3213(00)00774-4}{\path{doi}}.

\bibitem{MooRea-91}
G. Moore and N. Read, \emph{Nonabelions in the fractional quantum {H}all
  effect}, Nucl. Phys. B \textbf{360} (1991), no.~2, 362 -- 396,
  \href{http://dx.doi.org/10.1016/0550-3213(91)90407-O}{\path{doi}}.

\bibitem{Girvin-etal-90}
S.~M. Girvin, A.~H. MacDonald, M.~P.~A. Fisher, S.-J. Rey, and J.~P. Sethna,
  \emph{Exactly soluble model of fractional statistics}, Phys. Rev. Lett.
  \textbf{65} (1990), 1671--1674,
  \href{http://dx.doi.org/10.1103/PhysRevLett.65.1671}{\path{doi}}.

\bibitem{HanHerSimVie-16}
T.~H. Hansson, M. Hermanns, S.~H. Simon, and S.~F. Viefers, \emph{Quantum
  {H}all hierarchies}, arXiv e-prints, 2016,
  \href{http://arxiv.org/abs/1601.01697}{\path{arXiv:1601.01697}}.

\bibitem{ArdKedSto-05}
E. Ardonne, R. Kedem, and M. Stone, \emph{Filling the {B}ose sea: symmetric
  quantum {H}all edge states and affine characters}, Journal of Physics A:
  Mathematical and General \textbf{38} (2005), no.~3, 617,
  \href{http://dx.doi.org/10.1088/0305-4470/38/3/006}{\path{doi}}.

\bibitem{BerHal-08}
B.~A. Bernevig and F.~D.~M. Haldane, \emph{Model fractional quantum {H}all
  states and {J}ack polynomials}, Phys. Rev. Lett. \textbf{100} (2008), 246802,
  \href{http://dx.doi.org/10.1103/PhysRevLett.100.246802}{\path{doi}}.

\bibitem{Lundholm-13}
D. Lundholm, \emph{Anyon wave functions and probability distributions},
  IH{\'E}S preprint, IHES/P/13/25, 2013,
  \url{http://preprints.ihes.fr/2013/P/P-13-25.pdf}.

\bibitem{Chou-91b}
C. Chou, \emph{Multi-anyon quantum mechanics and fractional statistics}, Phys.
  Lett. A \textbf{155} (1991), no.~4, 245 -- 251,
  \href{http://dx.doi.org/10.1016/0375-9601(91)90477-P}{\path{doi}}.

\bibitem{Chou-91a}
C. Chou, \emph{Multianyon spectra and wave functions}, Phys. Rev. D \textbf{44}
  (1991), 2533--2547,
  \href{http://dx.doi.org/10.1103/PhysRevD.44.2533}{\path{doi}}.

\bibitem{MurLawBhaDat-92}
M.~V.~N. Murthy, J. Law, R.~K. Bhaduri, and G. Date, \emph{On a class of
  noninterpolating solutions of the many-anyon problem}, J. Phys. A: Math. Gen.
  \textbf{25} (1992), no.~23, 6163,
  \href{http://dx.doi.org/10.1088/0305-4470/25/23/013}{\path{doi}}.

\bibitem{ChoLeeLee-92}
M.~Y. Choi, C. Lee, and J. Lee, \emph{Soluble many-body systems with flux-tube
  interactions in an arbitrary external magnetic field}, Phys. Rev. B
  \textbf{46} (1992), 1489--1497,
  \href{http://dx.doi.org/10.1103/PhysRevB.46.1489}{\path{doi}}.

\bibitem{Mashkevich-96}
S. Mashkevich, \emph{{Finite-size anyons and perturbation theory}}, Phys. Rev.
  D \textbf{54} (1996), 6537--6543,
  \href{http://dx.doi.org/10.1103/PhysRevD.54.6537}{\path{doi}}.

\bibitem{Trugenberger-92}
C. Trugenberger, \emph{{The anyon fluid in the Bogoliubov approximation}},
  Phys. Rev. D \textbf{45} (1992), 3807--3817,
  \href{http://dx.doi.org/10.1103/PhysRevD.45.3807}{\path{doi}}.

\bibitem{Ouvry-94}
S. Ouvry, \emph{{$\delta$-perturbative interactions in the Aharonov-Bohm and
  anyons models}}, Phys. Rev. D \textbf{50} (1994), 5296--5299,
  \href{http://dx.doi.org/10.1103/PhysRevD.50.5296}{\path{doi}}.

\end{thebibliography}

\end{document}